\documentclass{article}
\usepackage{graphicx} 
\usepackage{amsmath,amssymb, amsthm}
\usepackage[left=4cm, right=4cm, top=3cm, bottom=3cm]{geometry}
\RequirePackage[colorlinks,citecolor=blue,urlcolor=blue]{hyperref}

\usepackage{bbm}
\usepackage{subcaption}
\usepackage{comment}
\usepackage{algorithm2e}
\RestyleAlgo{ruled}
\usepackage{multirow}
\usepackage{natbib}

\usepackage{bookmark}
\usepackage{setspace}
\usepackage{setspace}
\usepackage{siunitx} 
\newcommand{\real}{I\kern-0.37emR}
\newcommand{\bx}{\mathbf{x}}

\newcommand{\bY}{\mathbf{Y}}

\newcommand{\bA}{\mathbf{A}}

\newcommand{\bI}{\mathbf{I}}

\newcommand{\bb}{\mathbf{b}}

\newcommand{\btheta}{\boldsymbol{\theta}}

\newcommand{\hbtheta}{\widehat{\btheta}}

\newcommand{\bSigma}{\boldsymbol{\Sigma}}

\newtheorem{theorem}{Theorem}[section]

\newtheorem{assumption}{Assumption}[section]

\newtheorem{corollary}{Corollary}[section]
\newtheorem{remark}{Remark}[section]

\title{Generalized Spectral Testing with Sample Splitting}
\author{Yuxin Tao$^{1}$ \and Feiyu Jiang$^{2}$ \and Xiaofeng Shao$^{3}$
\thanks{1. Department of Statistics and Data Science, Southern University of Science and Technology; 2. Department of Statistics and Data Science, Fudan University; 3. Department of Statistics and Data Science \& Department of Economics, Washington University in St. Louis.
}}
\date{\today}

\begin{document}

\maketitle
\begin{abstract}Residual-based goodness-of-fit tests for parametric time-series models are often complicated by parameter-estimation effects, which can alter the limiting behavior of diagnostic statistics. We propose a sample-splitting generalized spectral test (in the spirit of \cite{escanciano2006goodness}) for assessing conditional mean specification in linear and nonlinear time-series models. The procedure estimates the model parameter on a fitting subsample and constructs a generalized spectral Cramér–von Mises statistic from residuals computed on a checking/testing subsample. The statistic aggregates pairwise conditional mean restrictions over all lags and is therefore bandwidth-free and free of truncation-lag selection. Under mild regularity conditions and a score-alignment condition, the residual-based process has the same limiting null distribution as the infeasible oracle process based on the true errors.  Although the resulting limiting law is still non-pivotal, it can be consistently approximated by a simple multiplier bootstrap that does not require generating bootstrap time series or re-estimating parameters. Such an oracle-equivalence property is in sharp contrast to the original full-sample test, for which parameter estimation contributes an additional first-order term to the limiting process, and requires re-estimating parameters in each bootstrapped sample. We further establish consistency of the proposed test against fixed alternatives and nontrivial power against local alternatives. Extensive simulations and real data analyses show that the proposed test controls size well, has comparable power, and delivers substantial computational savings in models where repeated estimation is costly.
\end{abstract}
\noindent%

\section{Introduction}

Diagnostic checking is a central step in time series analysis. Once a parametric dynamic model has been fitted, any subsequent inference, forecasting exercise, or structural interpretation rests on the maintained assumption that the conditional mean has been adequately specified. If the fitted model leaves systematic predictability in the residuals, parameter estimates may still be numerically stable, but the resulting inferential and forecasting conclusions can be misleading.  This motivates formal goodness-of-fit tests that examine whether the residuals satisfy the martingale difference property implied by a correctly specified conditional mean.

We consider a strictly stationary and ergodic time series $\{(Y_{t}, \mathbf{Z}_{t-1}^{\prime})\}_{t \in \mathbb{Z}}$  defined on the probability space $(\Omega, \mathcal{F}, P)$. Here $Y_t\in\mathbb R$ denotes the response variable, and
$\mathbf Z_{t-1}=(Y_{t-1},\mathbf X_{t-1}')'\in\mathbb R^m$ collects the lagged response and other predetermined explanatory covariates. Let
$\mathbf I_{t-1}=(\mathbf Z_{t-1}',\mathbf Z_{t-2}',\ldots)'$ denote the full information set available at time $t-1$, and let
$\mathcal F_{t-1}=\sigma(\mathbf I_{t-1})=\sigma(\mathbf Z_s:s\le t-1)$  be the $\sigma$-field generated by $\mathbf{I}_{t-1}$. 
Under the integrability of $Y_{t}$, the model can be written as
\begin{align*}
    Y_{t} = m (\mathbf{I}_{t-1})+\varepsilon_{t},
\end{align*}
where $m(\mathbf{I}_{t-1})=E[Y_{t} \mid \mathbf{I}_{t-1}]$ is the conditional mean almost surely (a.s.) given the conditioning set $\mathbf{I}_{t-1}$, and $\varepsilon_{t}=Y_{t}-E[Y_{t} \mid \mathbf{I}_{t-1}]$ is a martingale difference sequence (MDS) with respect to $\mathbf{I}_{t-1}$ as the error term. 

In this paper, we ask whether there exists a parametric family of functions $\mathcal{M}=\{f(\cdot, \boldsymbol{\theta}): \boldsymbol{\theta} \in \Theta \subset \mathbb{R}^{p} \}$, such that
\begin{align} \label{H0}
    H_{0}: m (\mathbf{I}_{t-1}) =f(\mathbf{I}_{t-1}, \boldsymbol{\theta}_{0}), \quad \text{for some }  \boldsymbol{\theta}_{0} \in \Theta,
\end{align}
or, equivalently,
\[
H_0:\quad E\{e_t(\boldsymbol{\theta}_0)\mid \mathbf{I}_{t-1}\}=0\ \text{a.s.},
\]
where $e_{t}(\boldsymbol{\theta})=Y_{t}-f(\mathbf{I}_{t-1},\boldsymbol{\theta})$ denotes the parameter induced model error.

Thus, goodness-of-fit can be formulated as a martingale difference hypothesis for the model errors. This formulation is especially natural in dynamic regression settings, where the relevant null is not merely the absence of linear autocorrelation, but the absence of any remaining conditional mean predictability.

Classical diagnostic tools are often based on residual autocorrelations and portmanteau statistics. These procedures remain attractive because they are simple, familiar, and easy to implement; see, among others, \citet{BoxPierce1970}, \citet{Hosking1980}, and \citet{LiMcLeod1981}. Their scope, however, is limited. Portmanteau-type checks are primarily geared toward linear serial dependence, and their standard reference distributions are often justified under innovation assumptions that are stronger than the null one would actually like to test. In particular, when the residuals are uncorrelated but still dependent, as may happen under conditional heteroskedasticity, the usual chi-squared calibration can fail; see \citet{Romano1996}. This mismatch is important in many nonlinear and heteroskedastic time series models, where the martingale difference null is the appropriate benchmark.

These concerns have led to a rich literature on omnibus specification testing for time series models. \citet{Hong1999} introduced generalized spectral methods for detecting serial dependence through empirical characteristic functions, and \citet{HongLee2005} extended this line of work to conditional mean models with conditional heteroskedasticity of unknown form. Particularly relevant for the present paper is the bandwidth-free generalized spectral test of \citet{escanciano2006goodness}, which provides a Cram\'er-von Mises type diagnostic for linear and nonlinear conditional mean models by integrating a continuum of pairwise moment restrictions. Subsequent work has broadened the scope of these ideas to joint and marginal checks for conditional mean and variance models \citep{Escanciano2008}, as well as to specification testing for conditional distribution models \citep{ChenHong2014}. More recently, \citet{Escanciano2024} has recast model checking in a broader conditional moment restriction framework. In a related direction, \cite{WangZhuShao2022} developed MDDM (Martingale Difference Divergence Matrix) based tests for the martingale difference hypothesis in multivariate time series models, and also provided a detailed discussion about the similarity and distinction between the MDDM-based approach and the one by \cite{escanciano2006goodness}.


A recurring technical difficulty in this literature is the estimation effect. Because the residuals used for diagnostic checking are computed from the same data that are used to estimate the unknown parameter, residual-based empirical processes do not, in general, behave as if the true innovations were observed. For many non-smoothing-based tests, this first-order estimation effect is asymptotically non-negligible, which leads to nonpivotal null distributions and complicates implementation. Existing remedies typically rely either on bootstrap procedures that repeatedly generate pseudo-data and refit the model, or on martingale-transform arguments designed to remove nuisance-parameter effects; see, for example, \citet{KoulStute1999}, \citet{KhmaladzeKoul2004}, and \citet{Bai2003}. While these approaches are theoretically powerful, they may be computationally demanding and can make routine model checking less convenient in practice.

Recent work by \citet{davis2025sample} shows that sample splitting offers an appealing alternative in time series goodness-of-fit problems. They generalize the half-sample splitting device proposed by \cite{durbin1976kolmogorov} from an independent data setting to the time series setting. The idea is to estimate the model parameters on one subsample and to assess goodness-of-fit on another, while allowing for a carefully chosen asymptotic overlap between the two parts of the sample. For diagnostics based on the autocorrelation function and the auto-distance correlation function, they show that the resulting residual-based statistics can have the same asymptotic behavior as in the infeasible oracle case, where the innovations are observed if the split-ratio between the fitting sample length and testing sample length is 1/2. This insight is important, but their framework is tailored to tests of residual serial independence and, in the portmanteau setting, it works under an independently and identically distributed (i.i.d.) innovation paradigm that is stronger than the martingale difference null. From the standpoint of conditional mean adequacy, it is therefore natural to ask whether the same sample-splitting idea can be combined with a diagnostic whose null hypothesis is aligned with the martingale difference property itself.

The goal of this paper is to answer that question. We develop a sample-splitting-based generalized spectral test for the goodness-of-fit of parametric conditional mean models. More specifically, we combine the integrated generalized spectral/Cram\'er-von Mises framework of \citet{escanciano2006goodness} with the sample-splitting strategy of \citet{davis2025sample}. One part of the sample is used for estimation, while a  checking/testing subsample is used to form residuals based on the fitted parameter and then to evaluate pairwise conditional mean restrictions of the form
\[
E\{e_t(\boldsymbol{\theta}_0)w(\mathbf Z_{t-j},\mathbf x)\}=0, \qquad j\geq 1,
\]
using appropriately chosen weighting functions $w(\mathbf Z_{t-j},\mathbf x)$ indexed by lagged explanatory variables $\mathbf{Z}_{t-j}$ and a continuum of values $\mathbf{x}$. In this way, the proposed procedure preserves the omnibus and bandwidth-free nature of generalized spectral diagnostics, while using sample splitting to separate model fitting from model assessment. As in \citet{escanciano2006goodness}, the resulting test targets pairwise conditional mean dependence rather than full joint conditional mean dependence on the entire past.

The paper makes several contributions. First, it develops a general split-sample generalized spectral testing framework for linear and nonlinear time series models with parametric conditional means. Second, under suitable regularity conditions and an overlap condition linking the fitting and checking subsamples, we show that the split-sample generalized spectral process has the same asymptotic null distribution as the corresponding infeasible oracle process based on the true errors. Thus, sample splitting eliminates the first-order parameter-estimation effect from the limiting null law. Third, although the limiting null distribution remains nonpivotal, it can be consistently approximated by a simple multiplier bootstrap applied directly to the split-sample residuals. This avoids generating bootstrap time series and re-estimating the model in each bootstrap replication, and therefore leads to a markedly lighter computational procedure. This is a major difference from \citet{escanciano2006goodness}, and the computational gain can be substantial. Similar to the latter paper, our procedure is also tuning parameter-free. 
Fourth, we establish consistency against fixed alternatives, derive nontrivial power against local alternatives, and verify the required conditions for important model classes, including ARMA and GARCH models. 

The proposed methodology complements several existing strands of literature. Relative to \citet{davis2025sample}, our focus is not residual serial independence per se, but the martingale difference hypothesis implied by a correctly specified conditional mean. Relative to \citet{escanciano2006goodness}, the contribution is not a new generalized spectral metric, but rather a new inferential device that uses sample splitting to simplify the effect of parameter estimation and to facilitate bootstrap implementation.   In this sense, the paper should be viewed as complementary to both earlier MDS testing work and the recent sample-splitting literature.

The rest of the paper is organized as follows. Section \ref{Sec2} proposes the split-sample generalized spectral process and the associated Cram\'er-von Mises statistic. Section \ref{Sec3} develops the asymptotic theory under the null hypothesis, fixed alternatives, and local alternatives. Section \ref{Sec4} presents the multiplier bootstrap approximation algorithm and its theoretical justification. 
Section \ref{Sec5} reports extensive simulation studies on the finite-sample performance of the proposed tests, including linear ARMA models, nonlinear volatility models, and threshold autoregressive models.
Section \ref{Sec6} presents empirical illustrations based on S\&P 500 dynamics and annual Wolf's sunspot data, demonstrating the validity and computational advantage of the proposed tests. 
Section \ref{Sec7} concludes.

\section{Generalized Spectral Tests with Sample Splitting} \label{Sec2}

This section constructs the sample-splitting generalized spectral test for \eqref{H0}.  The population restrictions and their integrated spectral aggregation follow the generalized spectral approach of \cite{escanciano2006goodness}: the null is represented by a continuum of pairwise conditional-mean restrictions, and the test measures whether the corresponding integrated spectral distribution is flat across the frequency index. The implementation, however, is different in an important way. In \cite{escanciano2006goodness}, the parameter is estimated from the full sample, and the same full sample is then used to construct the residual-based test statistic. Here, the parameter is estimated on a fitting subsample, whereas the generalized spectral process is formed from residuals on a checking subsample. 

We first introduce the generalized spectral test in \cite{escanciano2006goodness}. 
Under $H_0$, the model error satisfies $e_t(\boldsymbol{\theta}_0)=\varepsilon_t$ and hence has zero conditional mean given the full past. In particular, for every lag $j\ge 1$,
\begin{align} \label{H0-equiv}
  \gamma_{j} (\boldsymbol{\theta}_{0}) \equiv E [e_{t} (\boldsymbol{\theta}_{0} ) \mid \mathbf{Z}_{t-j} ]=0 
  \text { a.s., } \forall j\geq 1, \text { for some } \boldsymbol{\theta}_{0} \in \Theta \subset \mathbb{R}^{p}.
\end{align}
Thus, a correctly specified conditional mean implies an infinite collection of pairwise restrictions, one for each lagged vector $\mathbf Z_{t-j}$. To make these restrictions operational, let
$\{w(\mathbf Z_{t-j},\mathbf x):\mathbf x\in\Upsilon\subset\mathbb R^s\}$ be a family of bounded weighting functions such that the conditional restriction in \eqref{H0-equiv} is equivalently characterized by
\begin{align}  \label{H0-equiv2}
  \gamma_{j, w} (\mathbf{x}, \boldsymbol{\theta}_{0}) \equiv E[e_{t}(\boldsymbol{\theta}_{0}) w(\mathbf{Z}_{t-j}, \mathbf{x})]=0, \quad 
  \forall x \in \Upsilon \subset \mathbb{R}^{s}, j \geq 1.
\end{align}
As discussed later, usual choices of weight functions $w$ include $w (\mathbf{Z}_{t-j}, \mathbf{x})=\mathbbm{1} (\mathbf{Z}_{t-j} \leq \mathbf{x} )$ with $\mathbf{x} \in[-\infty, \infty]^{m}$, where $\mathbbm{1}(A)$ denotes the indicator function of the event $A$, and $w (\mathbf{Z}_{t-j}, \mathbf{x} )=\exp  (i \mathbf{x}^{\prime} \mathbf{Z}_{t-j} )$ with $\mathbf{x} \in \mathbb{R}^{m}$, where $i=\sqrt{-1}$ is the imaginary unit.

The sequence $\{\gamma_{j,w}(\cdot,\boldsymbol{\theta}_0):j\ge1\}$ measures departures from the pairwise martingale-difference implications across all lags. To aggregate these restrictions, define
$\gamma_{-j,w}(\cdot,\boldsymbol{\theta}_0)=\gamma_{j,w}(\cdot,\boldsymbol{\theta}_0)$ for $j\ge1$ and consider the generalized spectral transformation
\begin{align}
    f_{w} (u, \mathbf{x}, \boldsymbol{\theta}_{0} )=\frac{1}{2 \pi} \sum_{j=-\infty}^{\infty} \gamma_{j, w} (\mathbf{x}, \boldsymbol{\theta}_{0} ) e^{-i j u},
    \quad \forall u \in[-\pi, \pi], ~ \mathbf{x} \in \Upsilon.
\end{align}
Therefore, the tests are based on an integrated Fourier transform of the measures $ \{\gamma_{j, w} (\cdot, \boldsymbol{\theta}_{0} ) \}_{j=1}^{\infty}$. Under $H_0$, all nonzero-lag coefficients vanish, so
$$f_w(u,\mathbf x,\boldsymbol{\theta}_0) \equiv f_{0, w} (\mathbf{x}, \boldsymbol{\theta}_{0} )=(2 \pi)^{-1} \gamma_{0, w} (\mathbf{x}, \boldsymbol{\theta}_{0} ).$$
 Equivalently, the integrated generalized spectral distribution
\begin{align} \label{H_w}
    H_w(\lambda,\mathbf x,\boldsymbol{\theta}_0)
    &\equiv
    2\int_0^{\lambda\pi}f_w(u,\mathbf x,\boldsymbol{\theta}_0)\,du \nonumber \\
    &=
    \gamma_{0,w}(\mathbf x,\boldsymbol{\theta}_0)\lambda
    +2\sum_{j=1}^{\infty}
    \gamma_{j,w}(\mathbf x,\boldsymbol{\theta}_0)
    \frac{\sin j\pi\lambda}{j\pi},
    \quad \lambda\in[0,1],
\end{align}
reduces to 
$\gamma_{0,w}(\mathbf x,\boldsymbol{\theta}_0)\lambda$ under $H_0$. Therefore, any departure from this constant function implies a violation of at least one of the pairwise restrictions.

We now describe the proposed sample-splitting implementation. Suppose the observed sample is
$\{Y_t,\widehat{\mathbf I}_{t-1}\}_{t=1}^n$, where
$\widehat{\mathbf I}_{t-1}=(\mathbf Z_{t-1}',\mathbf Z_{t-2}',\ldots,\mathbf Z_0')'$ is the observed information set and may contain initial values. Let $f_n$ and $l_n$ be two non-decreasing sequences diverging with $n$. The first $f_n$ observations are used for estimation, and the last $l_n$ observations are used for model checking:
\begin{itemize}
    \item[(I)] the first $f_n$ data points, $\{Y_{t}, \widehat{\mathbf{I}}_{t-1} \}_{t=1}^{f_n}$, are the fitting sample to estimate the parameter $\btheta_0$ of the model;
    \item[(II)] the last $l_n$ data points, $\{Y_{t}, \widehat{\mathbf{I}}_{t-1} \}_{t=n-l_n+1}^{n}$, are the checking sample used to calculate the residuals based on the estimated parameter.
\end{itemize}

Let $\hbtheta_{f_n}$ denote the estimator computed from the fitting sample. On the checking sample, define the split-sample residuals
\[
\widehat e_t\equiv \widehat e_t(\hbtheta_{f_n})
=Y_t-f(\widehat{\mathbf I}_{t-1},\hbtheta_{f_n}),
\qquad t=n-l_n+1,\ldots,n.
\]
Note that, for lag $j$, the effective sample size is
$n_j=l_n-j+1.$ 
The empirical lag-$j$ weighted moment is
\begin{align*}
    \widehat{\gamma}_{j, w} (\mathbf{x}, \hbtheta_{f_n} )=\frac{1}{n_{j}} \sum_{t=n-l_n+j}^{n} \widehat{e}_{t} w (\mathbf{Z}_{t-j}, \mathbf{x} ), \quad 1\leq j \leq l_n,\quad  n_{j}=l_n-j+1 .
\end{align*}
Hence, the sample analogue of the generalized spectral distribution function \eqref{H_w} is given by
\begin{align*}
    \widehat{H}_{w} (\lambda, \mathbf{x}, \hbtheta_{f_n} )
= \widehat{\gamma}_{0, w} (\mathbf{x}, \hbtheta_{f_n} ) \lambda+2 \sum_{j=1}^{l_n} \widehat{\gamma}_{j, w} (\mathbf{x}, \hbtheta_{f_n} ) \Big(\frac{n_{j} }{l_n} \Big)^{1 / 2} \frac{\sin j \pi \lambda}{j \pi},
\end{align*}
where $ (n_{j} / l_n )^{1 / 2}$ is a finite-sample correction factor that does not affect the asymptotic theory and delivers a better finite-sample performance of the test procedure. 

Under $H_0$, the target function is
$H_w(\lambda,\mathbf x,\boldsymbol{\theta}_0)
=\gamma_{0,w}(\mathbf x,\boldsymbol{\theta}_0)\lambda$ and thus tests can be based on the discrepancy between $\widehat{H}_{w} (\lambda, \mathbf{x}, \hbtheta_{f_n} )$ and $\widehat{H}_{0, w} (\lambda, \mathbf{x}, \hbtheta_{f_n} ) \equiv \widehat{\gamma}_{0, w} (\mathbf{x}, \hbtheta_{f_n} ) \lambda$. Define  the centered process
\begin{align*}
S_{n,w}(\lambda,\mathbf x,\hbtheta_{f_n})
&\equiv
\left(\frac{l_n}{2}\right)^{1/2}
\{\widehat H_w(\lambda,\mathbf x,\hbtheta_{f_n})
-\widehat H_{0,w}(\lambda,\mathbf x,\hbtheta_{f_n})\}\\
&=
\sum_{j=1}^{l_n}
n_j^{1/2}\widehat{\gamma}_{j,w}(\mathbf x,\hbtheta_{f_n})
\frac{\sqrt2\sin j\pi\lambda}{j\pi}.
\end{align*}
Finally, we aggregate the process over $\Pi=[0,1]\times\Upsilon$ using a Cram\'er-von Mises norm:
\begin{align}
D_{n,w}^2(\hbtheta_{f_n})
&\equiv
\int_\Pi
|S_{n,w}(\lambda,\mathbf x,\hbtheta_{f_n})|^2
W(d\mathbf x)d\lambda \nonumber\\
&=
\sum_{j=1}^{l_n}
\frac{n_j}{(j\pi)^2}
\int_\Upsilon
|\widehat{\gamma}_{j,w}(\mathbf x,\hbtheta_{f_n})|^2
W(d\mathbf x),
\end{align}
where $W$ is an integrating distribution for the index $\mathbf x$. The test rejects $H_0$ for large values of $D_{n,w}^2(\hbtheta_{f_n})$. Because all lags in the checking sample are included with deterministic spectral weights, the procedure does not require choosing a truncation lag or bandwidth.

\section{Asymptotic Theory} \label{Sec3}
This section studies the large-sample behavior of the split-sample generalized spectral process introduced in Section \ref{Sec2}. The main message is that, under a simple score-alignment condition \eqref{key condition} for the data generating process and a half split-ratio condition on  the length of the fitting and checking samples, the residual-based process has the same null limit as the infeasible oracle process that would be computed from the true innovations. We then show that the corresponding Cram\'er-von Mises statistic is consistent against fixed pairwise alternatives and has nontrivial power against local alternatives.

To begin with, we introduce the notation used throughout the asymptotic theory. Let
\[
\mathbf{g}_{t}(\boldsymbol{\theta})
\equiv
\mathbf{g}(\mathbf{I}_{t-1},\boldsymbol{\theta})
=
\frac{\partial}{\partial\boldsymbol{\theta}}
f(\mathbf{I}_{t-1},\boldsymbol{\theta}),
\qquad
w_{t-j}(\mathbf{x})\equiv w(\mathbf Z_{t-j},\mathbf x).
\]
Recall that
\[
\varepsilon_t=Y_t-E(Y_t\mid\mathbf I_{t-1}),
\qquad
e_t(\boldsymbol{\theta})=Y_t-f(\mathbf I_{t-1},\boldsymbol{\theta}),
\]
so that under $H_0$, $e_t(\boldsymbol{\theta}_0)=\varepsilon_t$ a.s. Put
$\Pi=[0,1]\times\Upsilon$ and write
$\boldsymbol{\eta}=(\lambda,\mathbf x')'\in\Pi$.

The process $S_{n,w}(\boldsymbol{\eta},\hbtheta_{f_n})$ is viewed as a random element of the Hilbert space $L_2(\Pi,\nu)$, where $\nu$ is the product of the integrating measure $W$ on $\Upsilon$ and Lebesgue measure on $[0,1]$. For complex-valued functions $f,g \in L_2(\Pi,\nu)$, the inner product is
\[
\langle f,g\rangle
=
\int_{\Pi}f(\boldsymbol{\eta})g^c(\boldsymbol{\eta})\,d\nu(\boldsymbol{\eta})
=
\int_{\Pi}f(\lambda,\mathbf x)g^c(\lambda,\mathbf x)\,W(d\mathbf x)d\lambda,
\]
and $\|f\|=\langle f,f\rangle^{1/2}$. We equip $L_2(\Pi,\nu)$ with the Borel $\sigma$-field generated by this norm and write $\Longrightarrow$ for weak convergence in this space. For a mean-zero $L_2(\Pi,\nu)$-valued random element $Z$ with $E\|Z\|^2<\infty$, its covariance operator is
$C_Z(h)=E[\langle Z,h\rangle Z]$, $h\in L_2(\Pi,\nu)$.

Let
\[
\Psi_j(\lambda)=\frac{\sqrt{2}\sin(j\pi\lambda)}{j\pi},
\qquad
\mathbf b_j(\mathbf x,\boldsymbol{\theta}_0)
=E[w_{t-j}(\mathbf x)\mathbf g_t(\boldsymbol{\theta}_0)],
\]
and define the score-loading function
\[
\mathbf G_w(\boldsymbol{\eta})
\equiv
\mathbf G_w(\boldsymbol{\eta},\boldsymbol{\theta}_0)
=
\sum_{j=1}^{\infty}
\mathbf b_j(\mathbf x,\boldsymbol{\theta}_0)\Psi_j(\lambda).
\]
The covariance of the oracle generalized spectral limit is described through the quadratic form
\begin{align}  \label{sigma_h}
    \sigma_{h}^{2}= & \sum_{j=1}^{\infty} \sum_{k=1}^{\infty} E\Big[\varepsilon_{1}^{2} \int_{\Pi \times \Pi} h(\boldsymbol{\eta}_{1}) h^{c} (\boldsymbol{\eta}_{2}) w_{1-j}^{c}(\mathbf{x}) w_{1-k}(\mathbf{y}) 
\Psi_{j}(\lambda) \Psi_{k}(\varpi) d \nu (\boldsymbol{\eta}_{1}) d \nu (\boldsymbol{\eta}_{2})\Big],
\end{align}
where $h \in L_{2}(\Pi, \nu)$,
$\boldsymbol{\eta}_{1}=(\lambda, \mathbf{x}^{\prime})^{\prime}$, and
$\boldsymbol{\eta}_{2}= (\varpi, \mathbf{y}^{\prime})^{\prime}$.

\subsection{Asymptotic Null Distribution}

The following assumptions collect the necessary regularity conditions.
\begin{assumption}  \label{ass1}
(a). $\{Y_{t}, \mathbf{Z}_{t-1}\}_{t \in \mathbb{Z}}$ is a strictly stationary and ergodic process.\\
(b). $E[\varepsilon_{1}^{2}]<\infty$. 
\end{assumption}

\begin{assumption}   \label{ass2}
    The response function $f(\mathbf{I}_{t-1}, \cdot)$ is twice continuously differentiable on $\Theta$. The score $\mathbf{g}_{t}(\boldsymbol{\theta})$ is stationary, ergodic, and $\mathcal{F}_{t-1}$-measurable. There exists an integrable function $\mathbf{M}_t\in \mathcal{F}_{t-1}$ with $|\mathbf{g}(\mathbf{I}_{t-1}, \boldsymbol{\theta})| \leq \mathbf{M}_t$ for all $\theta\in\Theta$, and that $E\mathbf{M}_t^2<\infty.$
\end{assumption}

\begin{assumption}   \label{ass3}
    (a) The parameter space $\Theta$ is compact in $\mathbb{R}^{p}$, and the true parameter $\boldsymbol{\theta}_{0}$ belongs to the interior of $\Theta$. There exists a unique pseudo-true value $\boldsymbol{\theta}_{*} \in \Theta$ such that the estimator is consistent for $\boldsymbol{\theta}_*$ under both $H_0$ and fixed alternatives.\\
(b) Under $H_{0}$,   $\boldsymbol{\theta}_*=\boldsymbol{\theta}_0$, for estimators based on $\{Y_t,\mathbf{Z}_{t-1}\}_{t=1}^{f_n}$, as $f_n \to \infty $, the estimator ${\hbtheta}_{f_n}$ satisfies the expansion
\begin{align*}
    \sqrt{f_n}(\hbtheta_{f_n}-\boldsymbol{\theta}_{0})=\frac{1}{\sqrt{f_n}} \sum_{t=1}^{f_n} \mathbf{h}(Y_{t}, \mathbf{I}_{t-1}, \boldsymbol{\theta}_{0})+o_{p}(1),
\end{align*}
where $\mathbf{h}(\cdot)$ satisfies $E[\mathbf{h}(Y_{t}, \mathbf{I}_{t-1}, \boldsymbol{\theta}_{0}) | \mathcal{F}_{t-1}]=\mathbf{0}$, $\mathbf{L}(\boldsymbol{\theta}_{0})= E[\mathbf{h}(Y_{t}, \mathbf{I}_{t-1}, \boldsymbol{\theta}_{0}) \mathbf{h}^{\prime}(Y_{t}, \mathbf{I}_{t-1}, \boldsymbol{\theta}_{0})]$ exists and is positive definite.
By the martingale central limit theorem, the above implies
\begin{align*}
    \sqrt{f_n}(\hbtheta_{f_n}-\boldsymbol{\theta}_{0}) \stackrel{d}{\longrightarrow}
    \mathcal{N}(\mathbf{0}, \mathbf{L}(\boldsymbol{\theta}_{0})).
\end{align*}
\end{assumption}
\begin{assumption}   \label{ass4}
    The integrating function $W(\cdot)$ is a probability distribution function absolutely continuous with respect to Lebesgue measure. The weight function $w(\cdot)$ is such that the equivalence between \eqref{H0-equiv} and \eqref{H0-equiv2} holds and is uniformly bounded. Also, $w(\cdot)$ satisfies the uniform law of large numbers,
    \begin{align*}
        \sup _{\mathbf{x} \in \Upsilon_{c}} \bigg|n^{-1} \sum_{t=1}^{n} \zeta_{t} w(\boldsymbol{\xi}_{t}, \mathbf{x})-E \big[\zeta_{t} w(\boldsymbol{\xi}_{t}, \mathbf{x})\big]\bigg| \rightarrow 0, \quad \text { a.s.,}
    \end{align*}
whenever $\{(\zeta_{t}, \boldsymbol{\xi}_{t}^{\prime}), t=0, \pm 1, \ldots\}$ is a strictly stationary and ergodic process with $\zeta_{t} \in \mathbb{R}, \boldsymbol{\xi}_{t} \in \mathbb{R}^{m}, E|\zeta_{1}|<\infty$, and $\Upsilon_{c}$ is any compact subset of $\Upsilon \subset \mathbb{R}^{s}$.
\end{assumption}

\begin{assumption}   \label{ass5}
    The observed information set at period $t$, $\widehat{\mathbf{I}}_{t}$, may contain some assumed initial values and satisfies $\lim _{n \rightarrow \infty} \sum_{t=1}^{n}(E \sup _{\boldsymbol{\theta} \in \Theta}(f(\mathbf{I}_{t-1}, \boldsymbol{\theta})-f(\widehat{\mathbf{I}}_{t-1}, \boldsymbol{\theta}))^{2})^{1 / 2} <\infty$.
\end{assumption}
Assumptions \ref{ass1}--\ref{ass5} are inherited from \cite{escanciano2006goodness}, and are deliberately stated at a high level. Assumption \ref{ass1}-\ref{ass2} are standard but mild conditions on the data-generating process. Assumption \ref{ass3} (a) assumes a pseudo-true value under the alternative, which is standard in the theory of misspecified likelihood and other extremum estimation; see, e.g. \citet{newey1994large}, and \cite{white1982maximum,white1994estimation}.  
Assumption \ref{ass3} (b) assumes a Bahadur representation of the estimator 
on the fitting sample, and is satisfied by a broad class of estimators,
including commonly used estimators, such as (nonlinear) least squares, maximum
likelihood, and the (generalized) method of moments, under standard regularity
conditions.
Assumption \ref{ass4} is a high-level condition on the weight function. Such functions are often called ``generically comprehensively revealing'' functions; see, for example, \citet{bierens1997asymptotic} and \citet{stinchcombe1998consistent}. The boundedness of $w(\cdot)$ is mild, and the required uniform ergodic theorem for $\zeta_t w(\boldsymbol{\xi}_t,\bx)$ holds under standard uniform integrability conditions; see, e.g., Theorem 3.1 of \citet{ling2003asymptotic}. Typical examples include $w(\boldsymbol{\xi}_t,\bx)=\sin(\bx'\boldsymbol{\xi}_t)$, $w(\boldsymbol{\xi}_t,\bx)=[1+\exp(\bx'\boldsymbol{\xi}_t)]^{-1}$, $w(\bf{\xi}_t,\bx )=\exp( i \bx' \boldsymbol{\xi}_t)$,  and $w(\boldsymbol{\xi}_t,\bx)=\mathbb{I}(\boldsymbol{\xi}_t\leq \bx)$, see Lemma 1 of \citet{escanciano2006goodness} for further details. In the numerical studies, we focus on the latter two choices.
 Assumption \ref{ass5} allows the observed information set to approximate the infinite past, see similar assumptions in  \cite{HongLee2005} and \cite{WangZhuShao2022}. 

In addition to the above regularity conditions, we assume that the sample-splitting sequences are arbitrary non-decreasing sequences $\{(f_n, l_n )\}_{n \geq 1}$ which go to infinity with $n$, and $f_n/n \rightarrow \kappa_f$ and $l_n/n \rightarrow \kappa_l$ for some constants $0 < \kappa_f, \kappa_l \leq 1$.
Further, let $\kappa_{ra}\equiv \lim_{n \rightarrow \infty} l_n/f_n$ be the limiting ratio of sample-splitting, and $\kappa_{ov}=\lim_{n\rightarrow \infty} \max(0, (f_n+l_n-n)/f_n)$ be the limiting overlap coefficient.

\begin{theorem}  \label{thm1}
    Under Assumptions \ref{ass1}--\ref{ass5} and $H_{0}$, 
    if $\kappa_{ra} = 2 \kappa_{ov}$, and the following score-alignment condition holds, that is,
    \begin{equation}  \label{key condition}
        E\big[w_{t-j}(\mathbf{x}) \mathbf{g}_{t}(\btheta_0)\big]^{\prime} \mathbf{L}(\boldsymbol{\theta}_{0}) = E \Big[ \varepsilon_t  w_{t-j}(\mathbf{x}) \mathbf{h}(Y_{t}, \mathbf{I}_{t-1}, \boldsymbol{\theta}_{0}) \Big]^{\prime}, \quad \forall j \geq 1,
    \end{equation}
    then the process $S_{n, w}$ satisfies
    \begin{equation*}
        S_{n, w} \Longrightarrow S_w^0, \quad  \text{in} ~ L_{2}(\Pi, v),
    \end{equation*}
    where $S_{w}^{0}(\cdot)$ is a Gaussian process in $L_{2}(\Pi, v)$ with mean $\mathbf{0}$ and covariance operator $C_{S_{w}^{0}}$ satisfying $\sigma_{h}^{2}=\langle C_{S_{w}^{0}}(h), h \rangle, \forall h \in L_{2}(\Pi, v)$, and $\sigma_{h}^{2}$ is defined in \eqref{sigma_h}. 
    Moreover,
    by the continuous mapping theorem, we have
    \begin{align*}
        D_{n, w}^{2}(\hbtheta_{f_n}) \stackrel{d}{\longrightarrow} D_{\infty, w}^{2}(\boldsymbol{\theta}_{0})
        :=\int_{\Pi} |S_{w}^0 (\lambda, \mathbf{x}, \boldsymbol{\theta}_{0})|^{2} W(d \mathbf{x}) d \lambda.
    \end{align*}
\end{theorem}
Theorem \ref{thm1} states that the split-sample statistic has the oracle null limit, despite using estimated residuals. This is the main theoretical difference from the full-sample generalized spectral test of \citet{escanciano2006goodness}, where the same observations are used both to estimate the unknown parameter and to construct the residual marked process. In that full-sample construction, the first-order estimation effect is part of the limiting null distribution and must be accounted for in the approximation of critical values. By contrast, the present split-sample construction separates parameter fitting from model checking in a way that makes the residual-based process have the same limiting distribution as the oracle process based on the true errors.  
 
Although the limiting null distribution in Theorem \ref{thm1} is still nonstandard and requires approximation by bootstrap, the oracle-equivalence result has an important practical implication. The bootstrap procedure in Section \ref{Sec4} can apply multipliers directly to the split-sample residuals while keeping $\hbtheta_{f_n}$ fixed. Hence, unlike full-sample residual-based procedures that require re-generating data and re-estimating the model in each bootstrap replication, the proposed method avoids repeated refitting and can be substantially cheaper for nonlinear or numerically intensive models.

     \begin{remark}
          The equality $\kappa_{ra}=2\kappa_{ov}$ balances the amount of checking-sample information against the overlap between fitting and checking samples. A simple choice is $f_n=n/2$ and $l_n=n$. Condition \eqref{key condition} is a score-alignment condition: it matches the covariance between the residual moment process and the estimator influence function with the covariance generated by the model score.   A similar requirement appears in \citet{davis2025sample} for portmanteau statistics under the i.i.d. condition for the error process.  This paper extends it to the martingale difference setting.  More importantly, in Appendix \ref{secA.1} and \ref{secA.2},  we verify condition  \eqref{key condition}  for the ARMA Gaussian MLE and GARCH Quasi-MLE examples, respectively. 
     \end{remark}


\subsection{Consistency and Local Alternatives}

We proceed to prove the consistency of the proposed test under the fixed and local alternatives, respectively.
First, consider the fixed alternative:
\begin{align} \label{global H1}
    H_{a}: \quad Y_{t}=f (\mathbf{I}_{t-1}, \boldsymbol{\theta}_{*})+a_{t}+\varepsilon_{t},
\end{align}
where $\{a_{t}\}$ is strictly stationary and ergodic, with $Ea_{1}^2<\infty$, and importantly, for each $t \in \mathbb{Z}$, $a_{t}$ is $\mathcal{F}_{t-1}$-measurable and $P(a_t=0)<1$.
Theorem \ref{thm-global H1} shows the asymptotic behavior of $S_{n, w}$ under the global alternative $H_{a}$.
\begin{theorem}  \label{thm-global H1}
    Under Assumptions \ref{ass1}--\ref{ass5} and $H_{a}$ \eqref{global H1},
    \begin{align*}
        l_n^{-1/2}S_{n,w}(\cdot,\hbtheta_{f_n})
        \stackrel{p}{\longrightarrow}
        L_w(\cdot),
        \qquad \text{in } L_2(\Pi,\nu),
    \end{align*}
    where
    \begin{align*}
        L_w(\boldsymbol{\eta}):=\sum_{j=1}^\infty \varsigma_j(\mathbf{x})\Psi_j(\lambda),
        \quad
        \varsigma_j(\mathbf{x})=E[a_t w_{t-j}(\mathbf{x})].
    \end{align*}
\end{theorem}

\begin{corollary} \label{corollary1}
    Let $\Xi$ denote the class of alternatives $\{a_t\}$ such that, under $H_a$ \eqref{global H1}, there exists some $j\ge 1$ for which $\varsigma_j(\mathbf{x})\neq 0$ on a subset of $\Upsilon$ with positive W-measure.
    Then we have 
    \begin{align*}
        D_{n,w}^2(\hbtheta_{f_n}) \stackrel{p}{\longrightarrow}\infty,
    \end{align*}
    which means $D_{n,w}^2$ is consistent against all alternatives in $\Xi$.
\end{corollary}
Corollary \ref{corollary1} indicates that the test is consistent against alternatives that violate at least one of the pairwise conditional mean restrictions, namely $P(E[a_t|\boldsymbol{Z}_{t-j}]=0)<1$ for some $j\geq 1$. As in generalized spectral tests based on pairwise conditional moments, this is a large but not exhaustive class. It is indeed possible that $P(E[a_t| \mathbf I_{t-1}]=0)<1$ while all pairwise projections satisfy $E[a_t | \boldsymbol{Z}_{t-j}]=0$ a.s. for every fixed $j\geq 1$. Such alternatives involve higher-order or genuinely joint dependence on the past and are therefore not captured by a statistic built only from pairwise restrictions.

To further study the consistency properties of the proposed tests, we consider the local alternatives with a sequence of alternative hypotheses tending to the null at the parametric rate $l_n^{-1 / 2}$ as in \cite{escanciano2006goodness}:
\begin{equation} \label{local H1}
H_{a, n}: \quad Y_{t, n}=f(\mathbf{I}_{t-1}, \boldsymbol{\theta}_{0})+\frac{a_{t}}{\sqrt{l_n}}+\varepsilon_{t},
\end{equation}
where $\{a_{t}\}$ is the same as in $H_{a}$ \eqref{global H1}. Note that the rate $l_n$ is of the same order as $n$ in the half-splitting scheme.

An additional assumption is needed regarding the behavior of the estimator under these local alternatives.
\begin{assumption} \label{H1-ass1}
    The estimator $\hbtheta_{n}$ satisfies the asymptotic expansion under $H_{a, n}$ \eqref{local H1},
    $$
    \sqrt{f_n} (\hbtheta_{f_n}-\boldsymbol{\theta}_{0})= \boldsymbol{\xi}_{a}+\frac{1}{\sqrt{f_n}} \sum_{t=1}^{f_n} \mathbf{h} (Y_{t}, \mathbf{I}_{t-1}, \boldsymbol{\theta}_{0} )+o_{p}(1)
    $$
    where the function $\mathbf{h}(\cdot)$ is as in Assumption \ref{ass3} and $\boldsymbol{\xi}_{a} \in \mathbb{R}^{p}$.
\end{assumption} 

\begin{theorem} \label{thm-local H1}
     Under the sequence of alternative hypotheses \eqref{local H1} and Assumptions \ref{ass1}--\ref{H1-ass1}, if $\kappa_{ra} = 2 \kappa_{ov}$ and condition \eqref{key condition} holds, then
     \begin{align*}
         S_{n, w} \Longrightarrow S_{w}^0 + L_{w}(\cdot)-\sqrt{\kappa_{ra}} \mathbf{G}_{w}^{\prime}(\cdot) \boldsymbol{\xi}_{a},
     \end{align*}
     where $S_{w}^0$ and $L_{w}$ are the processes defined in Theorems \ref{thm1} and \ref{thm-global H1}, and $\kappa_{ra}\equiv \lim_{n \rightarrow \infty} l_n/f_n$. Moreover,
     \begin{align*}
         D_{n, w}^{2}(\hbtheta_{f_n}) 
         \stackrel{d}{\longrightarrow} \int_{\Pi} |S_{w}^0(\boldsymbol{\eta}, \boldsymbol{\theta}_{0})+L_{w}(\boldsymbol{\eta})-\sqrt{\kappa_{ra}} \mathbf{G}_{w}^{\prime}(\boldsymbol{\eta}) \boldsymbol{\xi}_{a}| ^{2} d \nu(\boldsymbol{\eta}).
     \end{align*}
\end{theorem}

Theorem \ref{thm-local H1} decomposes the local limit into three parts: the oracle Gaussian limit $S_w^0$, the direct local-alternative drift $L_w$, and the indirect drift caused by estimating $\boldsymbol{\theta}_0$ on the fitting sample. If $\boldsymbol{\xi}_{a} \neq \mathbf{0}$, local power may be reduced in directions that are aligned with the score $\mathbf{g}_{t}\left(\boldsymbol{\theta}_{0}\right)$. If $\boldsymbol{\xi}_{a}=\mathbf{0}$, the estimator has no first-order local drift, and the test has nontrivial power against all local alternatives in $\Xi$.

\begin{remark}
Compared with \cite{escanciano2006goodness}, the local power has two changes in the limiting process. 
Escanciano's statistic has the local limit $S_w+L_w(\cdot) -  \mathbf{G}_{w}^{\prime}(\cdot) \boldsymbol{\xi}_{a}$,
where the Gaussian component $S_w$ contains the first-order estimation effect. 
In contrast, the split-sample statistic has the limit $S_{w}^0 + L_{w}(\cdot)-\sqrt{\kappa_{ra}} \mathbf{G}_{w}^{\prime}(\cdot) \boldsymbol{\xi}_{a}$. 
Thus, sample splitting removes the estimation effect from the stochastic part of the limit, and changes the deterministic score-induced drift from $\mathbf{G}_{w}^{\prime} \boldsymbol{\xi}_{a}$ to $\sqrt{\kappa_{ra}} \mathbf{G}_{w}^{\prime} \boldsymbol{\xi}_{a}$.

Consequently, there is no uniform local power dominance. 
If the local alternative is nearly orthogonal to the score direction, so that $\mathbf{G}_{w}^{\prime} \boldsymbol{\xi}_{a}$ is small, then the two tests have essentially the same local drift $L_w$, while the split-sample test has the oracle stochastic component. 
In this case, sample splitting can yield higher local power. 
If $\mathbf{G}_{w}^{\prime} \boldsymbol{\xi}_{a}$ is not negligible, the power gain or loss depends on the sign and magnitude of $L_w(\cdot) - \sqrt{\kappa_{ra}} \mathbf{G}_{w}^{\prime}(\cdot) \boldsymbol{\xi}_{a}$ relative to $L_w(\cdot) - \mathbf{G}_{w}^{\prime}(\cdot) \boldsymbol{\xi}_{a}$.
\end{remark}

We will show in Section \ref{Sec5} that, in finite samples, the proposed tests exhibit empirical power comparable to that of the full-sample generalized spectral tests in most linear and nonlinear examples.

\section{Bootstrap Approximation} \label{Sec4}

The limiting null distribution in Theorem \ref{thm1} is nonstandard because the covariance operator of the limiting Gaussian process depends on the underlying data-generating process
(DGP) in a very complicated way. Hence, direct tabulation of the critical value is generally infeasible. However, compared with \cite{escanciano2006goodness}, a key observation is that the asymptotic distribution of $S_{n, w}(\boldsymbol{\eta}, \hbtheta_{f_n})$ is the same as  $S_w^0(\boldsymbol{\eta})$, which  does not involve any estimation effect of $\hbtheta_{f_n}$. We elaborate more on this point. 

In the full-sample generalized spectral test of \citet{escanciano2006goodness}, because the asymptotic distribution of the test therein involves the estimation effect, the bootstrap must reproduce this effect. Hence, \cite{escanciano2006goodness} applies the fixed-design wild bootstrap: one first generates bootstrap residuals, then constructs bootstrap responses, then { re-estimates} the model on the bootstrap sample to obtain $\boldsymbol{\theta}_n^*$, and finally recomputes bootstrap residuals using $\boldsymbol{\theta}_n^*$. This is why the procedure requires additional high-level assumptions on the bootstrap estimator. 

The split-sample construction in this paper changes the bootstrap problem. By Theorem \ref{thm1}, under the split-overlap and score-alignment conditions \eqref{key condition}, the residual-based process $S_{n,w}(\cdot,\hbtheta_{f_n})$ has the same limiting null distribution as the oracle process based on the true errors. Hence, the bootstrap only needs to mimic the oracle martingale fluctuation on the checking sample. This means that the fitting-sample estimator $\hbtheta_{f_n}$ can be kept fixed throughout the bootstrap, and no bootstrap responses are generated. This is the main computational advantage of the proposed procedure, especially for nonlinear models or models whose estimation requires cumbersome numerical optimization.

Specifically, let $\{V_t\}$ be a sequence of i.i.d. random variables, independent of the data, with zero mean, unit variance, and bounded support. Conditional on the observed sample, define
\begin{align}  \label{S_nw}
    S_{n,w}^*(\lambda, \bx, \hbtheta_{f_n} )
    =
    \sum_{j=1}^{l_n}
    n_{j}^{1 / 2}
    \widehat{\gamma}_{j,w}^*(\bx, \hbtheta_{f_n})
    \frac{\sqrt{2} \sin j \pi \lambda}{j \pi},
\end{align}
where
\begin{align} \label{gamma_jw}
    \widehat{\gamma}_{j,w}^*(\bx, \hbtheta_{f_n})
    =
    \frac{1}{n_j}
    \sum_{t=n-l_n+j}^n
    \widehat{e}_{t}(\hbtheta_{f_n})w_{t-j}(\bx)V_t,
    \qquad j=1,\ldots,l_n .
\end{align}
This is a multiplier version of the checking-sample moment process. It is related to the wild bootstrap of \citet{Wu1986bootstrap}, \citet{liu1988bootstrap}, and \citet{mammen1993bootstrap}, but unlike the full-sample fixed-design wild bootstrap used in \citet{escanciano2006goodness}, it does not require constructing $Y_t^*$ or re-estimating the parameter.

Common choices of $\{V_t\}$ include Mammen's two-point distribution,
\begin{align}\label{mammen}
    P\{V_t=(1-\sqrt{5})/2\}=\frac{1+\sqrt{5}}{2\sqrt{5}}, \quad 
    P\{V_t= (1+\sqrt{5})/2\}=1-\frac{1+\sqrt{5}}{2\sqrt{5}},
\end{align}
and the Rademacher distribution $P(V_t=1)=P(V_t=-1)=1/2$; see also \citet{stute1998bootstrap}, \citet{de1996bierens}, and \citet{mammen1993bootstrap}.


The complete bootstrap algorithm is summarized in Algorithm \ref{alg1}. 
\begin{algorithm}[t]
\caption{Bootstrap Algorithm for Generalized Spectral Tests with Sample Splitting.} \label{alg1}
\begin{enumerate}
    \item Estimate the original model using the first $f_n$ data points, $\{Y_{t}, \widehat{\mathbf{I}}_{t-1} \}_{t=1}^{f_n}$, and derive estimate $\hbtheta_{f_n}$.
    \item Use the last $l_n$ data points, $\{Y_{t}, \widehat{\mathbf{I}}_{t-1} \}_{t=n-l_n+1}^{n}$, to obtain the residuals $\widehat{e}_{t} (\hbtheta_{f_n} )$ based on the estimated parameter, and then calculate the test statistic $D_{n,w}^2$.
    \item Generate $\{V_t\}_{t=n-l_n+1}^n$, a sequence of i.i.d. random variables with zero mean, unit variance, and bounded support, and independent of the sample $\{Y_t, \widehat{\bI}_{t-1}\}_{t=1}^n$.
    \item Compute $\widehat{\gamma}_{j,w}^*(\bx, \hbtheta_{f_n})$ for $j=1, \ldots, l_n$ as in \eqref{gamma_jw}, with generated $\{V_t\}_{t=n-l_n+1}^n$. Then, derive $S_{n,w}^*(\boldsymbol{\eta}, \hbtheta_{f_n} )$ and $D_{n,w}^{* 2}$.
    \item Repeat Steps $3-4$ for $B$ times, and compute the empirical $(1-\alpha)$-th sample quantile of $B$ values of $D_{n,w}^{* 2}$, denoted as $D_{n,w,\alpha}^{* 2}$. The proposed test rejects the null hypothesis at the significance level $\alpha$ if $D_{n,w}^2>D_{n,w, \alpha}^{* 2}$.
\end{enumerate}
\end{algorithm}
The algorithm keeps $\hbtheta_{f_n}$ fixed in every bootstrap replication. Therefore, the computational cost of the bootstrap is dominated by recomputing the weighted residual moments, rather than by repeated model estimation. This distinction is minor for very simple linear models, but it can be substantial for nonlinear volatility models, and threshold models, see Section \ref{Sec5} for numerical evidence.

We then provide the validity of the bootstrap procedure. 
\begin{theorem} \label{thm2}
    Assume that Assumptions \ref{ass1}--\ref{ass5} and the conditions in Theorem \ref{thm1} hold. Under the null hypothesis $H_0$, under any fixed alternative hypothesis or under the local alternatives,
    \begin{align*}
        S_{n,w}^* \underset{*}{\Longrightarrow} \widetilde{S}_w \quad a.s.,
    \end{align*}
    where $\widetilde{S}_w$ is the same Gaussian process of Theorem \ref{thm1} but with $\theta_*$ replacing $\btheta_0$ and $\underset{*}{\Longrightarrow}$ denoting weak convergence a.s. under the bootstrap law (see \cite{gine1990bootstrapping}).
\end{theorem}

\section{Simulation Studies} \label{Sec5}

To examine the finite-sample performance of the proposed tests, we carry out simulation studies with different DGPs under the nulls and under the alternatives. 
Here we consider the univariate case, that is, $m=1$ with $Z_{t}=Y_{t}, t \in \mathbb{Z}$. 
Denote $D_{n, I}^{2}$ and $D_{n, C}^{2}$ to be the proposed CvM tests corresponding to $w(Y_{t}, x)=\mathbbm{1}(Y_{t} \leq x)$ and $w (Y_{t}, x)=\exp (i Y_{t} x)$, where $W(\cdot)$ is chosen to be the empirical CDF $F_{n}(\cdot)$ based on $\{Y_{t-1}\}_{t=n-l_n+1}^{n}$ and the CDF $\Phi(\cdot)$ of a standard normal random variable, respectively.
Then,
\begin{align*}
    D_{n, I}^{2}=\sum_{j=1}^{l_n} \frac{n_{j}}{l_n (j\pi)^{2}} \sum_{t=n-l_n+1}^{n} \widehat{\gamma}_{j, I}^{2} (Y_{t-1}, \hbtheta_{f_n}),
\end{align*}
with $\widehat{\gamma}_{j, I}(x, \hbtheta_{f_n})= (\widehat{\sigma}_{e} n_{j})^{-1} \sum_{k=n-l_n+j}^{n} \widehat{e}_{k} (\hbtheta_{f_n}) \mathbbm{1}(Y_{k-j} \leq x )$, $\widehat{\sigma}_{e}^{2}=l_n^{-1} \sum_{s=n-l_n+1}^{n} \widehat{e}_{s}^{2}(\hbtheta_{f_n})$, and
\begin{align*}
D_{n, C}^{2}=\sum_{j=1}^{l_n} \frac{1}{\widehat{\sigma}_{e}^{2} n_{j}(j \pi)^{2}} \sum_{t=n-l_n+j}^{n} \sum_{s=n-l_n+j}^{n} \widehat{e}_{t} (\hbtheta_{f_n}) \widehat{e}_{s} (\hbtheta_{f_n}) \exp \Big\{-\frac{1}{2}(Y_{t-j}-Y_{s-j})^{2}\Big\}.
\end{align*}
Other choices of integrating functions $W(\cdot)$ in $D_{n, C}^{2}$ can also be applied. 
Besides, we also compare the finite-sample performance with that of CvM tests in \cite{escanciano2006goodness} without sample splitting, which are denoted as $\widetilde{D}_{n,I}^2$ and $\widetilde{D}_{n,C}^2$ using FDWB approximation for the bootstrapped statistics.

For the simulation setup, 
we consider sample sizes  $n=200$ and $n=500$, and the nominal level of $5\%$ for each test.
The Monte Carlo experiments are repeated $1,000$ times, and in each replication, the initial $200$ burn-in observations of the generated processes are discarded to ensure stationarity.
The number of bootstrap replications is $B=500$, and the random weighting 
$\{V_{t}\}$ is generated by i.i.d. Mammen's distribution \eqref{mammen}.

\subsection{Linear ARMA Model}\label{sec:5.1}

First we consider the linear case, where the null model is the AR$(1)$ model as follows:
\begin{align}  \label{H0-AR(1)}
    H_0: Y_{t}=b Y_{t-1}+\varepsilon_{t}, \quad \text{where} \quad \varepsilon_{t} \overset{\text{i.i.d.}}{\sim} N(0,1).
\end{align}
We examine the goodness-of-fit of this model under the following DGPs:
\begin{enumerate}
  \item AR(1) model: $Y_{t}= 0.6 Y_{t-1}+\varepsilon_{t}$;
  \item AR(1) model with exponential centered noise (AR-EXP): $Y_{t}=0.6 Y_{t-1}+\bar{\epsilon}_{t}$, $\bar{\epsilon}_{t}=\epsilon_{t}-1, \epsilon_{t} \sim \exp (1)$;
  \item AR(1) model with heteroscedasticity (AR-HET): 
  $Y_{t}=0.6 Y_{t-1}+h_{t} \varepsilon_{t}$, $ h_{t}^{2}=0.1+ 0.1 Y_{t-1}^{2}$;
  \item AR(1) model plus a bilinear term (AR-BIL): 
  $Y_{t}=0.6 Y_{t-1}+ 0.1 Y_{t-1} \varepsilon_{t}+\varepsilon_{t}$;
  \item AR(2) model: $Y_{t} = 0.6 Y_{t-1}- 0.5 Y_{t-2}+\varepsilon_{t}$;
  \item ARMA(1,1) model: 
  $Y_{t}= 0.6 Y_{t-1}+0.9 \varepsilon_{t-1}+\varepsilon_{t}$;
  \item Bilinear model (BIL): $Y_{t}=0.6 Y_{t-1} + 0.7 \varepsilon_{t-1} Y_{t-2}+\varepsilon_{t}$;
  \item Nonlinear moving average model (NLMA): $Y_{t}=0.6 Y_{t-1}+0.7 \varepsilon_{t-1} \varepsilon_{t-2}+\varepsilon_{t}$;
  \item Threshold autoregressive model (TAR): $Y_{t}=0.6 Y_{t-1} + \varepsilon_{t}$ if $Y_{t-1}<1$ and $Y_{t}=-0.5 Y_{t-1} + \varepsilon_{t}$ if $Y_{t-1} \geq 1$;
  \item Sign autoregressive model (SIGN): $Y_{t}=\operatorname{sign} (Y_{t-1})+ 0.43 \varepsilon_{t}$, where $\operatorname{sign}(x)=\mathbbm{1}(x>0)-\mathbbm{1}(x<0)$;
  \item Temp map model (TEM MAP): $Y_{t}=\alpha^{-1} Y_{t-1}$ if $0 \leq Y_{t-1}<\alpha$ and $Y_{t}=(1-\alpha)^{-1} (1-Y_{t-1})$ if $\alpha \leq Y_{t-1} \leq 1$, where $\alpha=0.49999$ and $Y_{0}$ is generated from the uniform distribution on $[0,1]$;
  \item Nonlinear autoregressive model (NAR): $Y_{t}=0.6 Y_{t-1} + 0.7 \sin (0.3 \pi Y_{t-2})+\varepsilon_{t}$.
\end{enumerate}
DGPs 1--4 are introduced to examine the empirical size of the tests under the nulls, and DGPs 5--12 are used to examine the power under the alternatives.

\begin{table}[t]
\begin{center}
\caption{Empirical Size (\%) of Tests at $5\%$ with sample size $n=200$ and $n=500$, without intercept in linear $H_0$ \eqref{H0-AR(1)}.} \label{table1}
\begin{tabular}{l*{4}{S}}
\hline \hline
$n=200$ & {AR(1)} & {AR(1)-EXP} &  {AR(1)-HET} &  {AR(1)-BIL} \\
\hline
${D}_{n, I}^{2}$   & 4.5 & 6.4 & 6.9 & 8.7 \\
$D_{n, C}^{2}$   & 5.1 & 6.1 & 5.7 & 6.9 \\
$\widetilde{D}_{n, I}^{2}$   & 4.6 & 6.2 & 5.8 & 6.9 \\
$\widetilde{D}_{n, C}^{2}$   & 4.0 & 6.0 & 5.9 & 6.7 \\
\hline  \hline
$n=500$ & {AR(1)} & {AR(1)-EXP} & {AR(1)-HET} &  {AR(1)-BIL} \\
\hline 
$D_{n, I}^{2}$  & 5.7 & 4.0   & 4.4  & 7.4 \\
$D_{n, C}^{2}$   & 4.5 & 4.6  & 4.9  &5.2 \\
$\widetilde{D}_{n, I}^{2}$   & 6.0 & 5.0  & 5.0  &4.8 \\
$\widetilde{D}_{n, C}^{2}$   & 5.3 & 4.9  & 4.4  &3.7 \\
\hline  \hline
\end{tabular}
\end{center}
\end{table}

\begin{table}[t]
\centering
\caption{Empirical Power (\%) of Tests at $5 \%$ with sample size $n=200$ and $n=500$, without intercept in linear $H_0$ \eqref{H0-AR(1)}.}  \label{table2}
\begin{tabular}{l*{8}{S}}
\hline \hline
$n=200$ & {AR$(2)$} & {ARMA} & {BIL} & {NLMA} & {TAR} & {SIGN} & {TEM MAP} & {NAR} \\
\hline
${D}_{n, I}^{2}$   & 5.9 & 3.8 & 10.3 & 6.5 & 85.6 & 66.0 & 100.0 & 17.9 \\
$D_{n, C}^{2}$   & 22.2 & 3.2 & 37.3 & 15.8 & 98.4 & 93.9 & 100.0 & 44.5 \\
$\widetilde{D}_{n, I}^{2}$    & 7.9 & 3.4 & 9.5 & 7.5 & 97.5 & 71.0 & 100.0 & 23.3 \\
$\widetilde{D}_{n, C}^{2}$   & 29.4 & 4.0 & 27.1 & 13.0 & 98.8 & 98.6 & 100.0 & 47.9 \\
\hline \hline
$n=500$ & {AR$(2)$} &  {ARMA} & {BIL} & {NLMA} & {TAR} & {SIGN} & {TEM MAP} & {NAR} \\
\hline
${D}_{n, I}^{2}$   & 84.8 &  9.6 & 18.4 & 9.1 & 99.9 & 94.8 & 100.0 & 25.8 \\
$D_{n, C}^{2}$   & 98.7 &  19.7 & 82.1 & 35.9 & 100.0 & 100.0 & 100.0 & 88.0 \\
$\widetilde{D}_{n, I}^{2}$    & 92.5 &  10.3 & 12.8 & 7.5 & 100.0 & 99.2 & 100.0 & 30.4 \\
$\widetilde{D}_{n, C}^{2}$   & 99.3 &  23.8 & 69.3 & 31.6 & 100.0 & 100.0 & 100.0 & 95.6 \\
\hline  \hline
\end{tabular}
\end{table}

For the computation of all test statistics, we apply the usual least squares residuals under $H_0$, that is, $\widehat{e}_{t} (\hbtheta_{f_n})=Y_{t}-\widehat{b}_{f_n} Y_{t-1}$, for $t=n-l_n+1,\ldots,n$. 
Table \ref{table1} reports the empirical rejection probabilities (RPs) associated with DGPs 1--4 to examine the empirical size of the tests,
and Table \ref{table2} reports the empirical RP associated with DGPs 5-12 to examine the empirical power of the tests in the cases of $n=200$ and $n=500$.

Table \ref{table1} shows that the proposed split-sample generalized spectral testing framework attains empirical sizes close to the nominal 5\% level in finite samples under the linear AR(1) null models, and that these sizes are comparable to those of the generalized spectral tests in \cite{escanciano2006goodness}. 
This finding supports Theorem \ref{thm1}, which states that the proposed test has the same null limiting distribution as the oracle process.
Table \ref{table2} indicates that the proposed tests exhibit empirical power comparable to that of the full-sample generalized spectral tests, and that the power increases as the sample size grows from $n=200$ to $n=500$, which validates Theorem \ref{thm-global H1} for the consistency of the proposed tests.

\subsection{Nonlinear Volatility Model}

In the nonlinear case, we are interested in analyzing the size and power against misspecifications in the conditional variance.
We consider two common conditional heteroscedasticity models:

1. 
ARCH-type null model, e.g., ARCH$(1)$ model:
\begin{align}  \label{H0-ARCH}
    &H_0: y_t=\sigma_t \eta_t, ~
		\sigma_{t}^{2}= \omega + \phi y_{t-1}^{2}, 
        \quad \text{where} \quad \eta_{t} \overset{\text{i.i.d.}}{\sim} N(0,1).
\end{align}
and ARCH$(4)$ model:
\begin{align}  \label{H0-ARCH4}
    &H_0: y_t=\sigma_t \eta_t, ~
		\sigma_{t}^{2}= \omega + \sum_{i=1}^4 \phi_i y_{t-i}^{2}, 
        \quad \text{where} \quad \eta_{t} \overset{\text{i.i.d.}}{\sim} N(0,1).
\end{align}
2. 
GARCH-type null model, e.g., GARCH$(2,2)$ model:
\begin{align}  \label{H0-GARCH22}
    &H_0: y_t=\sigma_t \eta_t, ~
		\sigma_{t}^{2}= \omega + 
        \sum_{i=1}^2 \phi_i y_{t-i}^{2} + 
        \sum_{j=1}^2 \psi_j \sigma_{t-j}^{2}, 
        ~ \text{where} ~ \eta_{t} \overset{\text{i.i.d.}}{\sim} N(0,1).
\end{align}
Recall that under the null, $Y_t$ corresponds to $y_t^2$, $E[Y_t | \mathbf{I}_{t-1}] = \sigma_t^2$, and the MDS $\varepsilon_t = Y_t - E[Y_t | \mathbf{I}_{t-1}] = \sigma_t^2 (\eta_t^2 -1)$.
We examine the goodness-of-fit of this model under the following DGPs, where for all cases $\eta_{t} \overset{\text{i.i.d.}}{\sim} N(0,1)$:
\begin{enumerate}
  \item ARCH(1) model: $y_t=\sigma_t \eta_t, ~
		\sigma_{t}^{2}=0.9 + 0.1  y_{t-1}^{2}$;
  \item ARCH(2) model: $y_t=\sigma_t \eta_t, ~
		\sigma_{t}^{2}=0.9 + 0.1 y_{t-1}^{2} + 0.8 y_{t-2}^{2}$;
  \item ARCH(4) model: $y_t=\sigma_t \eta_t, ~
		\sigma_{t}^{2}=0.9 + 0.1 y_{t-1}^{2} + 0.2 y_{t-2}^{2} + 0.2 y_{t-3}^{2} + 0.1 y_{t-4}^{2}$;
  \item ARCH(5) model: $y_t=\sigma_t \eta_t, ~
		\sigma_{t}^{2}=0.9 + 0.1 y_{t-1}^{2} + 0.2 y_{t-2}^{2} + 0.2 y_{t-3}^{2} + 0.1 y_{t-4}^{2} + 0.3 y_{t-5}^{2}$;
 \item GARCH(1,1) model: $y_t=\sigma_t \eta_t, ~
		\sigma_{t}^{2}=0.01 + 0.29  y_{t-1}^{2} + 0.7 \sigma_{t-1}^{2}$;
  \item GARCH(2,2) model: $y_t=\sigma_t \eta_t, ~
		\sigma_{t}^{2}=0.1 + 0.2 y_{t-1}^{2} + 0.2 y_{t-2}^{2} + 0.3 \sigma_{t-1}^{2} + 0.1 \sigma_{t-2}^{2}$;
  \item EGARCH(1,1) model: 
  $y_t=\sigma_t \eta_t, ~
		\log \sigma_{t}^{2}=0.01 + 0.9 \log \sigma_{t-1}^{2} + 0.3(|\eta_{t-1}| - (2/\pi)^{1/2}) - 0.8 \eta_{t-1}$;
  \item Stochastic volatility (SV) model: $y_t=\sigma_t \eta_t, ~
		\sigma_{t}^{2}= 0.1 y_{t-1}^2 + \exp (0.98 \log \sigma_{t-1}^2 + v_t)$;
  \item Bilinear model (BIL): $y_{t}=0.8 \eta_{t-1} y_{t-1}+\eta_{t}$;
  \item Logistic map (LM): $y_{t}= 4 y_{t-1} (1 - y_{t-1})$, where $y_{0}$ is generated from the uniform distribution on $[0,1]$;
  \item Nonlinear moving average model (NLMA): $y_{t}=0.8 \eta_{t-1}^2+\eta_{t}$.
\end{enumerate}
The DGPs for comparisons are the same as those in \cite{Escanciano2008}.

To compute the test statistics, we apply the quasi-MLE for estimation of unknown parameters $\btheta= (\omega, \phi_1, \ldots, \phi_p, \psi_1, \ldots, \psi_q)' \in \mathbb{R}^{p+q+1}$, and define the residuals to be $\widehat{e}_{t} (\hbtheta_{f_n})=Y_{t}-\sigma_t^2(\hbtheta_{f_n})$, for $t=n-l_n+1,\ldots,n$, where $\hbtheta_{f_n}$ is the quasi-MLE for $\btheta_0$ defined in Section \ref{secA.2} using the first $f_n$ samples.

\begin{table}[t]
\centering
\caption{Empirical Size and Power (\%) of Tests at $5\%$ with sample size $n=200$ and $n=500$ for nonlinear $H_0$ \eqref{H0-ARCH} under ARCH$(1)$ model with DGP 1.}  \label{table12}
\begin{tabular}{l *{8}{S}}
\hline \hline
$200$ & {ARCH$(1)$} & {ARCH$(2)$} & {GARCH} & {EGARCH} & {SV} & {BIL} & {LM} & {NLMA} \\
\hline
${D}_{n, I}^{2}$   & 3.3 & 77.6 & 55.5 & 41.4 & 100.0 & 4.2 & 100.0 & 13.5 \\
$D_{n, C}^{2}$   & 3.5 & 82.6 & 51.6 & 42.8 & 100.0 & 9.4 & 4.2 & 28.1 \\
$\widetilde{D}_{n, I}^{2}$    & 2.5 & 74.0 & 56.9 & 32.9 & 100.0 & 13.8 & 100.0 & 50.9 \\
$\widetilde{D}_{n, C}^{2}$   & 2.5 & 72.3 & 24.9 & 29.2 & 100.0 & 16.9 & 100.0 & 50.1 \\
\hline \hline
$500$ & {ARCH$(1)$} & {ARCH$(2)$} & {GARCH} & {EGARCH} & {SV} & {BIL} & {LM} & {NLMA} \\
\hline
${D}_{n, I}^{2}$   & 3.1 & 96.7 & 71.2 & 59.2 & 100.0 & 5.6 & 100.0 & 39.6 \\
$D_{n, C}^{2}$   & 3.9 & 96.7 & 64.8 & 60.4 & 100.0 & 22.9 & 40.3 & 68.2 \\
$\widetilde{D}_{n, I}^{2}$    & 1.8 & 86.7 & 91.3 & 56.2 & 100.0 & 31.6 & 100.0 & 93.3 \\
$\widetilde{D}_{n, C}^{2}$   & 3.1 & 86.4 & 53.2 & 54.9 & 100.0 & 36.6 & 100.0 & 89.7 \\
\hline \hline
\end{tabular}
\end{table}

For the ARCH$(1)$ null \eqref{H0-ARCH}, Table \ref{table12} reports the empirical rejection probabilities (RPs) associated with DGPs 1--2, 5, and 7--11 to examine the empirical size and power of tests at 5\% with sample size $n=200$ and $n=500$.
It shows that the proposed split-sample tests achieve empirical sizes close to the nominal $5\%$ level, with size performance comparable to that of the original full-sample generalized spectral tests. Their empirical power is also broadly comparable across the nonlinear alternatives considered, and it increases as the sample size rises from $n=200$ to $n=500$.

For the ARCH$(4)$ null \eqref{H0-ARCH4}, Table \ref{table14} reports the empirical rejection probabilities (RPs) under the DGPs 3--5 and 7--11.
It indicates that the proposed split-sample tests again maintain empirical sizes close to the nominal level and have power comparable to that of the corresponding full-sample tests. Moreover, in several cases, the split-sample procedures outperform the original tests in finite samples, both in terms of size accuracy and empirical power. The power improvement with larger sample sizes is also evident when changing the sample size from $n=200$ to $n=500$.

\begin{table}[t]
\centering
\caption{Empirical Size and Power (\%) of Tests at $5 \%$ with sample size $n=200$ and $n=500$ for nonlinear $H_0$ \eqref{H0-ARCH4} under ARCH$(4)$ model with DGP 3.}  \label{table14}
\begin{tabular}{l *{8}{S}}
\hline \hline
$200$ & {ARCH$(4)$} & {ARCH$(5)$} & {GARCH} & {EGARCH} & {SV} & {BIL} & {LM} & {NLMA} \\
\hline
${D}_{n, I}^{2}$   & 6.6 & 13.5 & 19.7 & 24.9 & 100.0 & 3.0 & 100.0 & 8.7 \\
$D_{n, C}^{2}$   & 6.2 & 12.2 & 18.3 & 23.3 & 100.0 & 9.8 & 15.5 & 21.2 \\
$\widetilde{D}_{n, I}^{2}$    & 1.5 & 3.1 & 4.7 & 19.0 & 100.0 & 10.3 & 82.3 & 31.0 \\
$\widetilde{D}_{n, C}^{2}$   & 1.3 & 2.9 & 8.3 & 16.0 & 100.0 & 13.1 & 74.4 & 29.9 \\
\hline \hline
$500$ & {ARCH$(4)$} & {ARCH$(5)$} & {GARCH} & {EGARCH} & {SV} & {BIL} & {LM} & {NLMA} \\
\hline
${D}_{n, I}^{2}$   & 4.8 & 16.5 & 18.1 & 54.2 & 100.0 & 6.4 & 100.0 & 19.7 \\
$D_{n, C}^{2}$   & 5.2 & 14.7 & 17.9 & 50.6 & 100.0 & 25.5 & 43.7 & 52.3 \\
$\widetilde{D}_{n, I}^{2}$    & 3.2 & 7.6 & 8.2 & 57.3 & 100.0 & 30.5 & 99.3 & 85.6 \\
$\widetilde{D}_{n, C}^{2}$   & 4.0 & 5.8 & 10.8 & 51.4 & 100.0 & 37.5 & 96.4 & 82.1 \\
\hline \hline
\end{tabular}
\end{table}

For the GARCH$(2,2)$ null \eqref{H0-GARCH22}, Table \ref{table10} reports the empirical rejection probabilities (RPs) under the DGPs 4 and 6--11.
It shows a similar pattern that the proposed split-sample tests deliver empirical sizes close to the nominal $5\%$ level and exhibit empirical power broadly comparable to that of the original full-sample tests. 
In several scenarios, they even improve upon the original procedures in terms of finite-sample size control or rejection power. 
In the other settings, the power generally increases with the sample size.

\begin{table}[t]
\centering
\caption{Empirical Size and Power (\%) of Tests at $5 \%$ with sample size $n=200$ and $n=500$ for nonlinear $H_0$ \eqref{H0-GARCH22} under GARCH$(2,2)$ model with DGP 6.}  \label{table10}
\begin{tabular}{l*{8}{S}}
\hline \hline
$n=200$ & {GARCH$(2,2)$} & {ARCH$(5)$} & {EGARCH} & {SV} & {BIL} & {LM} & {NLMA} \\
\hline
${D}_{n, I}^{2}$   & 6.8 & 14.4 &  30.8 & 100.0  & 3.6 & 100.0 & 8.4 \\
$D_{n, C}^{2}$   & 7.6 & 13.7  & 30.0  & 100.0 & 11.6 & 4.4 & 21.2   \\
$\widetilde{D}_{n, I}^{2}$   & 3.6 & 4.3 & 21.6 & 100.0 & 12.0 & 98.8 & 36.0   \\
$\widetilde{D}_{n, C}^{2}$   & 2.0  & 5.4 & 20.8 & 100.0 & 16.0 & 93.6 & 34.8 \\
\hline \hline
$n=500$ & {GARCH$(2,2)$} & {ARCH$(5)$} & {EGARCH} & {SV} & {BIL} & {LM} & {NLMA} \\
\hline
${D}_{n, I}^{2}$   & 5.7 & 7.0 & 50.6 & 100.0 & 6.4 & 100.0 & 27.1  \\
$D_{n, C}^{2}$   & 5.9 & 6.9 & 49.1 & 100.0 & 25.2 & 36.5 & 58.7   \\
$\widetilde{D}_{n, I}^{2}$   & 2.4 & 6.1 & 60.6 & 100.0 & 29.2 & 99.9 & 84.2    \\
$\widetilde{D}_{n, C}^{2}$   & 4.2 & 4.4 & 54.2 & 100.0 & 35.2 & 99.1 & 80.3   \\
\hline \hline
\end{tabular}
\end{table}

\subsection{Threshold Autoregressive Model}\label{sec:5.3}

We proceed to study more complex nonlinear cases, for example, the threshold autoregressive (TAR) model.
Simulation studies show that our proposed test still has consistently good size and power against model misspecifications.
We consider the following three-regime self-exciting TAR(1,1,1; $d=1$) process as null model, and set the threshold variable to be of lag 1:
\begin{align}  \label{H0-TAR(1)}
    H_0: Y_{t}=\begin{cases}
\phi_1 Y_{t-1}+\varepsilon_t, & \text{if } Y_{t-1} < r_1, \\
\phi_2 Y_{t-1}+\varepsilon_t, & \text{if } r_1 \leq Y_{t-1} < r_2, \\
\phi_3 Y_{t-1}+\varepsilon_t, & \text{if } Y_{t-1} \geq r_2,
\end{cases}
\quad \text{where} \quad \varepsilon_{t} \overset{\text{i.i.d.}}{\sim} N(0,1).
\end{align}
We examine the goodness-of-fit of this model under the following DGPs:
\begin{enumerate}
  \item Self-exciting TAR(1,1,1; $d=1$) model: 
  $Y_{t}=0.6 Y_{t-1} + \varepsilon_{t}$ if $Y_{t-1}<-1$; 
  $Y_{t}= -0.5 Y_{t-1} + \varepsilon_{t}$ if $-1 \leq Y_{t-1}<1$;
  and $Y_{t}=0.3 Y_{t-1} + \varepsilon_{t}$ if $Y_{t-1} \geq 1$;
  \item AR(2) model: $Y_{t} = 0.6 Y_{t-1}- 0.5 Y_{t-2}+\varepsilon_{t}$;
  \item ARMA(1,1) model: 
  $Y_{t}= 0.6 Y_{t-1}+0.9 \varepsilon_{t-1}+\varepsilon_{t}$;
  \item Bilinear model (BIL): $Y_{t}=0.6 Y_{t-1} + 0.7 \varepsilon_{t-1} Y_{t-2}+\varepsilon_{t}$;
  \item Nonlinear moving average model (NLMA): $Y_{t}=0.6 Y_{t-1}+0.7 \varepsilon_{t-1} \varepsilon_{t-2}+\varepsilon_{t}$;
  \item Sign autoregressive model (SIGN): $Y_{t}=\operatorname{sign} (Y_{t-1})+ 0.43 \varepsilon_{t}$, where $\operatorname{sign}(x)=\mathbbm{1}(x>0)-\mathbbm{1}(x<0)$;
  \item Temp map model (TEM MAP): $Y_{t}=\alpha^{-1} Y_{t-1}$ if $0 \leq Y_{t-1}<\alpha$ and $Y_{t}=(1-\alpha)^{-1} (1-Y_{t-1})$ if $\alpha \leq Y_{t-1} \leq 1$, where $\alpha=0.49999$ and $Y_{0}$ is generated from the uniform distribution on $[0,1]$;
  \item Nonlinear autoregressive model (NAR): $Y_{t}=0.6 Y_{t-1} + 0.7 \sin (0.3 \pi Y_{t-2})+\varepsilon_{t}$.
\end{enumerate}

To compute the test statistics, for each fixed pair $(r_1,r_2)'$, we first estimate $(\phi_1,\phi_2,\phi_3)'$ by least squares. 
We then select the optimal threshold pair 
$(r_1,r_2)'$ over a grid by minimizing the least squares criterion. 
The residuals $\widehat e_t(\widehat{\boldsymbol\theta}_{f_n})$ are then obtained by plugging in the estimator $\hbtheta_{f_n} = (\widehat\phi_1,\widehat\phi_2,\widehat\phi_3,\widehat r_1,\widehat r_2)'$ using the first $f_n$ samples, for $t=n-l_n+1,\ldots,n$.

Table \ref{table16} reports the empirical sizes and powers (\%) of the tests at the nominal $5\%$ level for the null hypothesis $H_0$ of the three-regime TAR$(1,1,1;d=1)$ model in \eqref{H0-TAR(1)}, under DGPs 1--8 and sample sizes $n=200$. 
We omit the case $n=500$ for now because the full-sample generalized spectral tests are quite time-consuming for the three-regime TAR model, and $n=200$ is sufficient to illustrate the validity of the proposed tests.
The results show that the proposed split-sample tests again maintain empirical sizes close to the nominal level and exhibit power comparable to that of the corresponding full-sample tests.

\begin{table}[t]
\centering
\caption{Empirical size and power (\%) of $5\%$ tests for the null hypothesis $H_0$ of three-regime TAR$(1,1,1;d=1)$ model in \eqref{H0-TAR(1)}, under the DGPs 1-8 and sample sizes $n=200$.
}  \label{table16}
\begin{tabular}{l*{8}{S}}
\hline \hline
$n=200$ & {TAR} & {AR$(2)$} & {ARMA$(1,1)$} & {BIL} & {NLMA} & {SIGN} & {TEM MAP} & {NAR} \\
\hline
${D}_{n, I}^{2}$               & 6.8 & 14.4 &  8.2   & 10.0  & 12.1 & 40.2 & 63.2 & 21.4 \\
$D_{n, C}^{2}$               & 6.6 & 35.2  &  15.4  & 26.1 & 17.1 & 47.0 & 14.2 & 36.4  \\
$\widetilde{D}_{n, I}^{2}$   & 6.0 & 65.0  &   8.0  & 11.4 & 8.0 & 54.2 & 100.0 & 8.5  \\
$\widetilde{D}_{n, C}^{2}$   & 6.4 & 75.6  & 9.8  & 15.0 & 9.1 & 64.9 & 100.0 & 25.6 \\
\hline \hline
\end{tabular}
\end{table}

\subsection{Computational Time}

We now collect the computational time comparisons across the simulation designs. The reported times are averaged over five replications for one Monte Carlo experiment.  
Table \ref{table:time_summary} reports the corresponding results for the null models in Sections \ref{sec:5.1}-\ref{sec:5.3}, namely the AR(1), ARCH(1), ARCH(4), GARCH(2,2) and three-regime TAR(1,1,1) models.

\begin{table}[t]
\begin{center}
\caption{Computational time (s) for one Monte Carlo experiment for each test with five replications under different null models.}
\label{table:time_summary}
\setlength{\tabcolsep}{4pt}
\begin{tabular}{lccccc}
\hline \hline
Model & {AR$(1)$} & {ARCH$(1)$} & {ARCH$(4)$} & {GARCH$(2,2)$} & {\shortstack{TAR$(1,1,1)$}} \\
\hline
\multicolumn{6}{c}{$n=200$} \\
\hline
$D_{n,I}^{2}$             & 3.60 & 4.49 & 3.19 & 3.69 & 6.16 \\
$D_{n,C}^{2}$             & 4.15 & 6.71 & 3.14 & 4.30 & 8.13 \\
$\widetilde{D}_{n,I}^{2}$ & 3.60 & 7.23 & 40.18 & 98.90 & 3649.35 \\
$\widetilde{D}_{n,C}^{2}$ & 4.15 & 9.69 & 40.12 & 103.55 & 3609.34 \\
\hline
\multicolumn{6}{c}{$n=500$} \\
\hline
$D_{n,I}^{2}$             & 69.17 & 66.04 & 59.89 & 70.97 & 341.95 \\
$D_{n,C}^{2}$             & 61.58 & 97.96 & 44.54 & 59.57 & 409.39 \\
$\widetilde{D}_{n,I}^{2}$ & 70.27 & 71.83 & 127.86 & 133.56 & 334211.33 \\
$\widetilde{D}_{n,C}^{2}$ & 61.72 & 103.86 & 112.50 & 119.15 & 332321.66 \\
\hline \hline
\end{tabular}
\end{center}
\end{table}

The timing results highlight where sample splitting is most useful computationally. In the linear AR case, the proposed and full-sample generalized spectral tests have similar computational costs, because least squares estimation is inexpensive even when repeated inside the bootstrap. A similar pattern appears for the ARCH$(1)$ null, where repeated estimation is still relatively light. The advantage becomes much clearer for higher-order volatility models. Under the ARCH$(4)$ and GARCH$(2,2)$ nulls, the full-sample procedures require repeated bootstrap re-estimation of a larger nonlinear model, whereas the proposed tests estimate the model only once and then apply multipliers directly to the split-sample residuals.

The computational gain is most pronounced for the three-regime TAR model. There, each full-sample bootstrap replication requires another threshold search and model refit, which is costly even for moderate sample sizes. By contrast, the split-sample procedure keeps the fitted parameter fixed throughout the bootstrap. As a result, the proposed tests can be much faster while maintaining accurate size and comparable power performance.

\section{Empirical Application} \label{Sec6}

\subsection{S\&P 500 Dynamics}

The dynamic behavior of stock returns has been a central topic in financial econometrics, since the specification of the conditional mean and conditional variance directly affects risk measurement, portfolio allocation, and asset pricing. 
In particular, it is important to distinguish whether the serial dependence in financial returns is mainly driven by the conditional mean, the conditional variance, or both.

We use the same dataset as in the empirical study of  \cite{Escanciano2008}, namely the log differences of the S\&P 500 daily stock index in a sample period from January 1, 1988 to May 28, 1993. Following the literature, we delete the last 10\% of the observations, leaving 1138 observations. We examine the following three model specifications:
\begin{enumerate}
    \item an AR(1) model with conditional homoskedastic errors;
    \item a GARCH(1,1) model without the AR(1) component in the conditional mean;
    \item an AR(1)-GARCH(1,1) model, where both the conditional mean and conditional variance are modeled.
\end{enumerate}
The bootstrap $p$-values are reported in Table \ref{table:SP500}. 
The first four columns $D_{n, I}^{2}$ and $D_{n, C}^{2}$ give the results of the proposed sample-splitting tests in \eqref{S_nw}, while the last four columns $\widetilde{D}_{n,I}^2$ and $\widetilde{D}_{n,C}^2$ report the corresponding full-sample tests in \cite{escanciano2006goodness, Escanciano2008}. 
For each fitted model, we consider both marginal and joint tests for the conditional mean and conditional variance, denoted by the subscripts $m$ and $v$, respectively.
The number of bootstrap replications is $B=500$.

Table \ref{table:SP500} shows that the sample-splitting tests yield conclusions very similar to those from the full-sample tests. 
For the AR(1) model with conditional homoskedastic errors, the $p$-values are generally large for both the conditional mean and conditional variance tests. This indicates that the linear AR(1) specification captures the conditional mean adequately, and there is no strong evidence against conditional homoskedasticity. 

When the AR(1) component of the mean is neglected and only a GARCH(1,1) model is fitted, the conditional mean tests strongly reject the null hypothesis.
This suggests that ignoring the linear dependence in the conditional mean leads to misspecification. 
In contrast, for the AR(1)-GARCH(1,1) model, the conditional mean appears to be reasonably specified, but the conditional variance tests produce zero $p$-values. 
Hence, the GARCH specification for the conditional variance is rejected. 
These findings are consistent with the conclusions in the existing study of \cite{hong2003diagnostic} and \cite{Escanciano2008}: the S\&P 500 returns over this period are better described by a linear AR(1) mean structure with conditional homoskedastic errors, rather than by an AR(1)-GARCH(1,1) model.

\begin{table}[]
\begin{center}
\caption{Bootstrap $p$-values for the sample-splitting and full-sample tests of conditional mean and conditional variance specifications based on the S\&P 500 daily stock index.}
\label{table:SP500}
\begin{tabular}{c|*{4}{c}|*{4}{c}}
\hline \hline
\rule{0pt}{2.8ex} \multirow{2}{*}{Models} & \multicolumn{4}{c|}{Sample-splitting Tests}
 & \multicolumn{4}{c}{Full-sample Tests} \\ [2pt]
\cline{2-5} \cline{6-9}
 & {$\rule{0pt}{2.8ex} D_{n,I,m}^{2}$} 
 & {$D_{n,C,m}^{2}$} 
 & {$D_{n,I,v}^{2}$} 
 & {$D_{n,C,v}^{2}$} 
 & {$\widetilde{D}_{n,I,m}^{2}$} 
 & {$\widetilde{D}_{n,C,m}^{2}$} 
 & {$\widetilde{D}_{n,I,v}^{2}$} 
 & {$\widetilde{D}_{n,C,v}^{2}$} \\ [2pt]
\hline
1 & 0.278 & 0.100 & 0.680 & 0.868 & 0.193 & 0.130 & 0.468 & 0.123 \\
2 & 0.000 & 0.000 & 0.000 & 0.000 & 0.000 & 0.000 & 0.000 & 0.000 \\
3 & 0.280 & 0.108 & 0.000 & 0.000 & 0.145 & 0.153 & 0.000 & 0.000 \\
\hline
\end{tabular}
\end{center}
\end{table}

Overall, this empirical application confirms the usefulness of the sample-splitting tests. 
Although the sample-splitting procedure uses only part of the sample for estimation and the remaining part for testing, it delivers conclusions comparable to those from the full-sample tests.

\subsection{Sunspot Data for the TAR Model}

To further illustrate the robustness of the proposed test and its computational advantage, we consider a classical dataset in the TAR literature, that is, the annual Wolf's sunspot data from 1700 to 1979.
This dataset has been studied in \cite{tong1978threshold}, \cite{tsay1989testing} and the references therein.
The series consists of 280 observations and is well known for its asymmetric cyclical behavior.
Let $\{Y_t\}_{t=1}^{280}$ denote the annual sunspot series.

Following the model specification in \cite{tsay1989testing}, we fit $\{Y_t\}_{t=1}^{280}$ using the following three-regime TAR model:
\begin{align} \label{tar-sunspot}
Y_t =
\begin{cases}
\phi_{10} + \sum_{i=1}^{11}\phi_{1i}Y_{t-i}+\varepsilon_t,
& \text{if } Y_{t-2}<r_1, \\[1.2ex]
\phi_{20} +\sum_{i=1}^{10}\phi_{2i}Y_{t-i}+\varepsilon_t,
& \text{if } r_1\leq Y_{t-2}<r_2, \\[1.2ex]
\phi_{30} +\sum_{i=1}^{10}\phi_{3i}Y_{t-i}+\varepsilon_t,
& \text{if } Y_{t-2}\geq r_2,
\end{cases}
\end{align}
where $Y_{t-2}$ is selected as the threshold variable, and the AR orders are 11, 10, and 10, respectively.
The unknown autoregressive coefficients are estimated by least squares for each fixed pair of threshold values, and the thresholds are selected by minimizing the least squares criterion.
In \cite{tsay1989testing}, the martingale difference assumption that $E(\varepsilon_t|\mathbf{I}_{t-1}) = 0$ is imposed but not formally tested. It motivates us to apply the generalized spectral test to examine the MDS assumption.

Table \ref{table:sunspot} reports the bootstrap $p$-values for different models. 
The first four columns $D_{n, I}^{2}$ and $D_{n, C}^{2}$ give the results of the proposed sample-splitting tests in \eqref{S_nw}, while the last four columns $\widetilde{D}_{n,I}^2$ and $\widetilde{D}_{n,C}^2$ report the corresponding full-sample tests of \cite{escanciano2006goodness, Escanciano2008}. 
The number of bootstrap replications is $B=500$.

The results show strong evidence in favor of the three-regime TAR(11,10,10) specification for the annual sunspot data.
The large bootstrap $p$-values indicate that the residuals from the fitted TAR model do not exhibit significant departures from the martingale difference property. Moreover, the sample-splitting tests lead to conclusions very similar to those obtained from the full-sample tests.
Meanwhile, the computational gain is substantial. The proposed sample-splitting tests are more than 20 times faster than the full-sample tests.
This gain comes from avoiding repeated model re-estimation and threshold search in the bootstrap procedure, which is computationally expensive for TAR models with multiple regimes and threshold parameters.

For comparison, we also fit the dataset $\{Y_t\}_{t=1}^{280}$ using a constant mean model:  
\begin{align*}
    Y_t = \phi_0 + \varepsilon_t,
\end{align*}
or a linear AR(p) model:  
\begin{align*}
    Y_t = \phi_{0} + \sum_{i=1}^{p}\phi_{i}Y_{t-i}+\varepsilon_t,
\end{align*}
with $p=5,10$, and 11. 
The testing results are summarized in Table \ref{table:sunspot}.
From this table, we can find that the constant model and AR(5) model are strongly rejected by all tests. 
Meanwhile, the AR(10) model is rejected by most tests at the 10\% significance level, and the bootstrap $p$-values for the AR(11) model remain relatively small compared with those for the three-regime TAR model. 
Overall, the proposed sample-splitting tests produce conclusions consistent with the full-sample tests and suggest that the constant and linear AR specifications are inadequate for the annual sunspot series.

\begin{table}[]
\begin{center}
\caption{Bootstrap $p$-values for the sample-splitting and full-sample tests based on the annual sunspot data from 1700 to 1979 under different models, as well as computational time (s) for testing the three-regime TAR model.}
\label{table:sunspot}
\begin{tabular}{c|*{2}{c}|*{2}{c}}
\hline \hline
\rule{0pt}{2.8ex} \multirow{2}{*}{Models} 
& \multicolumn{2}{c|}{Sample-splitting Tests}
& \multicolumn{2}{c}{Full-sample Tests} \\ [2pt]
\cline{2-3} \cline{4-5}
& {$\rule{0pt}{2.8ex} D_{n,I}^{2}$} 
& {$D_{n,C}^{2}$} 
& {$\widetilde{D}_{n,I}^{2}$} 
& {$\widetilde{D}_{n,C}^{2}$} \\ [2pt]
\hline
\rule{0pt}{2.5ex} TAR(11,10,10) & 0.706 & 0.838 & 1.000 & 1.000 \\
Constant & 0.000 & 0.000 & 0.000 & 0.000 \\
AR(5) & 0.000 & 0.000 & 0.000 & 0.004 \\
AR(10) & 0.082 & 0.100 & 0.096 & 0.930 \\
AR(11) & 0.136 & 0.176 & 0.156 & 0.958 \\ [1pt]
\hline
\rule{0pt}{2.5ex} Time for TAR model & 11.11 & 10.73 & 265.28 & 265.02 \\ [1pt]
\hline
\end{tabular}
\end{center}
\end{table}

\section{Conclusion and Discussion} \label{Sec7}

This paper develops a sample-splitting generalized spectral test for diagnostic checking of parametric time-series models. The proposed procedure combines the omnibus, bandwidth-free structure of generalized spectral tests with a split-sample estimation strategy that removes the first-order effect of parameter estimation under suitable split-overlap and score-alignment conditions. The resulting statistic targets pairwise conditional mean dependence in residuals and therefore provides a diagnostic checking for violations of the martingale-difference implication of a correctly specified conditional mean model.

A practical advantage of the method is that critical values can be obtained from a multiplier bootstrap applied directly to the split-sample residuals. The bootstrap avoids generating artificial time series and avoids re-estimating the model in each replication, which can lead to substantial computational savings for nonlinear or numerically intensive models. The simulation results suggest that the test has a reliable size and competitive power across linear, heteroskedastic, and nonlinear dynamic specifications.

Several extensions are worth pursuing. First, the same sample-splitting principle may be useful for broader conditional moment restriction tests, including recent Gaussian-process-based model checks \citep{Escanciano2024}. Second, the choice of weight function and integrating measure may affect finite-sample power, and a systematic comparison of indicator, characteristic-function, and data-adaptive weights would be valuable. Third, while the present paper focuses on pairwise conditional mean restrictions, extensions toward diagnostics for conditional quantile restrictions, higher-order or joint conditional mean and variance dependence remain important directions for future research.

\newpage

\clearpage
\appendix

\section{Verification of Main Assumptions}

\subsection{ARMA Model} \label{secA.1}

We show in this subsection that the ARMA model satisfies Assumptions \ref{ass1}--\ref{ass5} and condition \eqref{key condition} with Gaussian maximum likelihood estimation (MLE). 
Consider a causal and invertible ARMA($p, q$) process:
\begin{align}  \label{ARMA}
    Y_t=\sum_{k=1}^p \phi_{k} Y_{t-k}+\varepsilon_t+\sum_{\ell=1}^q \psi_{\ell} \varepsilon_{t-\ell}, \quad t \in \mathbb{Z},
\end{align}
where $\btheta= (\phi_1, \ldots, \phi_p, \psi_1, \ldots, \psi_q) \in \mathbb{R}^{p+q}$. 
Here, $\{\varepsilon_j\}$ are a series of independent and identically distributed (i.i.d.) random variables with mean zero and variance $\sigma^2<\infty$.
Suppose the observations $\{Y_t\}$ are from model \eqref{ARMA} with true parameter $\btheta_0$.
Define the backshift operator as $B$ such that $B Y_t=Y_{t-1}$.
The above process can be rewritten as
\begin{align*}
    \phi(B) Y_t = \psi(B) \varepsilon_t,
\end{align*}
where $\phi(z)=1-\sum_{k=1}^p \phi_k z^k$ is the autoregressive (AR) polynomial, $\psi(z)=1+\sum_{\ell=1}^q \psi_{\ell} z^{\ell}$ is the moving average (MA) polynomial.
Under invertibility, $\phi(z) / \psi(z)$ has a power series expansion, that is,
\begin{align*}
    \frac{\phi(z)}{\psi(z)}=\sum_{k=0}^{\infty} \pi_k(\btheta) z^k = 1 + \sum_{k=1}^{\infty} \pi_k(\btheta) z^k.
\end{align*}
Then the parametric error at time $t$ for an estimate $\hbtheta$ can be expressed as
\begin{align}  \label{arma-error}
    e_t(\hbtheta) = Y_t + \sum_{k=1}^{\infty} \pi_k(\hbtheta) Y_{t-k}, \quad t \in \mathbb{Z},
\end{align}
and the fitted residuals given the observed sample $\{Y_{t}, \widehat{\mathbf{I}}_{t-1} \}_{t=1}^{n}$ are
\begin{align*}
    \widehat{e}_t(\hbtheta) = Y_t + \sum_{k=1}^{t-1} \pi_k(\hbtheta) Y_{t-k}, \quad 1 \leq t \leq n.
\end{align*}

For the ARMA$(p,q)$ model, we consider the conventional pseudo maximum likelihood estimator $\widehat{\btheta}_n$ based on the Gaussian likelihood, defined as
\begin{align*}
    \widehat{\btheta}_n := \arg \max_{\btheta \in \Theta}  L (\btheta, \sigma^2 )
    =  \arg \max_{\btheta \in \Theta}  \sigma^{-n}(\operatorname{det} \bSigma)^{-1 / 2} \exp \Big(\frac{-1}{2 \sigma^2} \bY_n^{\mathrm{T}} \bSigma^{-1} \bY_n\Big),
\end{align*}
where $\bY_n=(Y_1, \ldots, Y_n)^{\prime}$ and the covariance $\bSigma(\btheta):=\operatorname{Var}(\bY_n) / \sigma^2$ is independent of $\sigma^2$. 
\cite{brockwell1991time} shows that this Gaussian MLE $\hbtheta_n$ is consistent and asymptotically normal even for non-Gaussian $\varepsilon_t$.

From \eqref{arma-error}, we obtain that under $H_0$, 
$Y_t = - \sum_{k=1}^{\infty} \pi_k(\btheta_0) Y_{t-k} + \varepsilon_t$, $t \in \mathbb{Z}$. 
Then
\begin{align*}
    &f(\mathbf{I}_{t-1}, \btheta_0) = E[Y_{t} | \mathbf{I}_{t-1}] = - \sum_{k=1}^{\infty} \pi_k(\btheta_0) Y_{t-k}, \\
    &\mathbf{g}_{t}(\boldsymbol{\theta}_0)= \frac{ \partial f(\mathbf{I}_{t-1}, \boldsymbol{\theta}_0) }{\partial \boldsymbol{\theta} }
    = - \sum_{k=1}^{\infty} \frac{  \partial \pi_k(\btheta_0)}{\partial \boldsymbol{\theta}} Y_{t-k}.
\end{align*}
The next corollary shows that the conditions of Theorems \ref{thm1}--\ref{thm-local H1} hold.

\begin{corollary} \label{corollary2}
  The causal and invertible ARMA($p, q$) process \eqref{ARMA} satisfies Assumptions \ref{ass1}--\ref{ass5} and condition \eqref{key condition} with Gaussian maximum likelihood estimation (MLE). Specifically,
  \begin{align*}
	E \Big[ \varepsilon_t  w_{t-j}(\mathbf{x}) \mathbf{h}(Y_{t}, \mathbf{I}_{t-1}, \boldsymbol{\theta}_{0}) \Big]^{\prime} 
	=
	E\big[w_{t-j}(\mathbf{x}) \mathbf{g}_{t}(\boldsymbol{\theta}_0)\big]^{\prime} \mathbf{L}(\boldsymbol{\theta}_{0}), \quad \forall j \geq 1.
\end{align*}
Hence, if $\kappa_{ra} = 2 \kappa_{ov}$, by Theorems \ref{thm1}--\ref{thm-local H1}, the process $S_{n, w}$ satisfies
    \begin{equation*}
        S_{n, w} \Longrightarrow S_w^0, \quad  \text{in} ~ L_{2}(\Pi, v),
    \end{equation*}
    and the proposed test is consistent under the alternatives.
\end{corollary}

\begin{proof}

Now we verify the hypotheses for the sample-splitting-based generalized spectral tests for the causal and invertible ARMA($p, q$) process \eqref{ARMA}.
Assumptions \ref{ass1}, \ref{ass4} are basic assumptions, and Assumptions \ref{ass2}--\ref{ass3} and \ref{ass5} are verified in the Section A.2 of Supplementary Material of \cite{davis2025sample}, invoking the uniform exponential decay of the coefficients in the AR$(\infty)$ representation of ARMA($p, q$) process implied by causality and invertibility.
It remains to verify condition \eqref{key condition}.

Recall that from \eqref{arma-error}, under $H_0$, 
$Y_t = - \sum_{k=1}^{\infty} \pi_k(\btheta_0) Y_{t-k} + \varepsilon_t$, $t \in \mathbb{Z}$. 
Then
\begin{align*}
    &f(\mathbf{I}_{t-1}, \btheta_0) = E[Y_{t} | \mathbf{I}_{t-1}] = - \sum_{k=1}^{\infty} \pi_k(\btheta_0) Y_{t-k}, \\
    &\mathbf{g}_{t}(\boldsymbol{\theta}_0)= \frac{ \partial f(\mathbf{I}_{t-1}, \boldsymbol{\theta}_0) }{\partial \boldsymbol{\theta} }
    = - \sum_{k=1}^{\infty} \frac{  \partial \pi_k(\btheta_0)}{\partial \boldsymbol{\theta}} Y_{t-k}.
\end{align*}
By Theorem 8.11.1 and (8.11.5)--(8.11.8) of \cite{brockwell1991time}, we have
\begin{align*}
   \sqrt{n} (\hbtheta_n-\btheta_0 ) 
   &= \frac{1}{\sigma^{2}} \mathbf{L}(\boldsymbol{\theta}_{0}) \frac{1}{\sqrt{n}}  \sum_{t=1}^n \mathbf{g}_t (\btheta_0 ) \varepsilon_t + o_P(1) \\
   & \stackrel{d}{\longrightarrow} \mathcal{N}(0,  \mathbf{L}(\boldsymbol{\theta}_{0})).
\end{align*}
It implies that
\begin{align*}
    \mathbf{h}(Y_{t}, \mathbf{I}_{t-1}, \boldsymbol{\theta}_{0}) = \frac{1}{\sigma^{2}} \mathbf{L}(\boldsymbol{\theta}_{0}) \mathbf{g}_t (\btheta_0 ) \varepsilon_t.
\end{align*}
Thus, it can be shown that
\begin{align*}
	E \Big[ \varepsilon_t  w_{t-j}(\mathbf{x}) \mathbf{h}(Y_{t}, \mathbf{I}_{t-1}, \boldsymbol{\theta}_{0}) \Big]^{\prime} 
	&= E \Big[ \varepsilon_t  w_{t-j}(\mathbf{x}) \frac{1}{\sigma^{2}} \mathbf{L}(\boldsymbol{\theta}_{0}) \mathbf{g}_t (\boldsymbol{\theta}_0 ) \varepsilon_t \Big]^{\prime} \\
	&= \frac{1}{\sigma^{2}} E\Big[ E \Big( \varepsilon_t  w_{t-j}(\mathbf{x})  \mathbf{g}_t (\boldsymbol{\theta}_0 ) \varepsilon_t \Big| \mathcal{F}_{t-1} \Big) \Big]^{\prime} \mathbf{L}(\boldsymbol{\theta}_{0}) \\
	&= \frac{1}{\sigma^{2}}  E(\varepsilon_t^2)  E\big[ w_{t-j}(\mathbf{x})  \mathbf{g}_t (\boldsymbol{\theta}_0 ) \big]^{\prime} \mathbf{L}(\boldsymbol{\theta}_{0})\\
	&=
	E\big[w_{t-j}(\mathbf{x}) \mathbf{g}_{t}(\boldsymbol{\theta}_0)\big]^{\prime} \mathbf{L}(\boldsymbol{\theta}_{0}), \quad \forall j \geq 1,
\end{align*}
which verifies the condition \eqref{key condition}.
   
\end{proof}

\subsection{GARCH Model} \label{secA.2}

In this section, we show that the GARCH model also satisfies \eqref{key condition} with quasi maximum likelihood estimation. Consider a stationary GARCH$(p,q)$ process:
\begin{align}\label{garch}
    \bigg\{
	\begin{array}{l}
		y_t=\sigma_t \eta_t, \\
		\sigma_{t}^{2}=\omega+ \sum_{i=1}^p \phi_i  y_{t-i}^{2} + \sum_{j=1}^q \psi_j \sigma_{t-j}^{2},
	\end{array} \bigg.
	~ t \in \mathbb{Z},
\end{align}
where $\omega>0$, $\phi_i \geq 0 ~ (i=1, \dots, p)$,  $\psi_j \geq 0 ~ (i=1, \dots, q)$, and
$\{\eta_t\}$ is a sequence of i.i.d. random variables with mean $0$ and variance $1$.
Then the unknown parameters are $\btheta= (\omega, \phi_1, \ldots, \phi_p, \psi_1, \ldots, \psi_q)' \in \mathbb{R}^{p+q+1}$. 
Suppose the observations $\{y_t\}$ are from model \eqref{garch} with true parameter $\btheta_0 = (\omega_0, \phi_{01}, \ldots, \phi_{0p}, \psi_{01}, \ldots, \psi_{0q})'$.
From \cite{bougerol1992stationarity}, model \eqref{garch} is strictly stationary \textit{if and only if} the top Lyapunov exponent is strictly negative, i.e., $\gamma(\bA_0)<0$, where
\begin{align*}
    &\gamma(\bA_0) := \lim_{t \rightarrow \infty} \frac{1}{t} E(\log \| A_{0,t} A_{0,t-1} \dots A_{0,1}\| ), \\
    &A_{0 ,t}=\left(\begin{array}{cccccccccc}
\phi_{01} \eta_t^2 & & \cdots & &  \phi_{0 p} \eta_t^2 & \psi_{01} \eta_t^2 & & \cdots & &  \psi_{0 q} \eta_t^2 \\
1 & 0 & \cdots & & 0 & 0 & & \cdots & & 0 \\
0 & 1 & \cdots & & 0 & 0 & & \cdots & & 0 \\
\vdots & \ddots & \ddots & & \vdots & \vdots & \ddots & \ddots & & \vdots \\
0 & & \cdots & 1 & 0 & 0 &  & \cdots  & 0 & 0 \\
\phi_{01} & & \cdots & & \phi_{0 p} & \psi_{01} & & \cdots & & \psi_{0 q} \\
0 & & \cdots & & 0 & 1 & 0 & \cdots & & 0 \\
0 & & \cdots & & 0 & 0 & 1 & \cdots & & 0 \\
\vdots & \ddots & \ddots & & \vdots & \vdots & \ddots & \ddots & & \vdots \\
0 & & \cdots & 0 & 0 & 0 & & \cdots & 1 & 0
\end{array}\right) \in \mathbb{R}^{(p+q)\times (p+q)}.
\end{align*}
Define the parameter space $\Theta \in \mathbb{R}_{+}^{1+p+q}$.
By iterations and under the stationarity condition, we can write
\begin{align*}
    \sigma_t^2(\btheta) 
    = \omega + \sum_{i=1}^p \phi_i  y_{t-i}^{2} + \sum_{j=1}^q \psi_j \sigma_{t-j}^{2}(\btheta) 
    = c_0(\btheta) + \sum_{i=1}^{\infty} c_i(\btheta) y_{t-i}^2, \quad t \in \mathbb{Z},
\end{align*}
for some suitably defined functions $c_i(\cdot)$; see \cite{berkes2003garch} for details. 
Given the finite sample $\{Y_{t}, \widehat{\mathbf{I}}_{t-1} \}_{t=1}^{n}$, the conditional variance can be approximated by
\begin{align*}
    \widetilde{\sigma}_t^2(\btheta) = c_0 (\btheta) + \sum_{i=1}^{t-1} c_i (\btheta) y_{t-i}^2.
\end{align*}
By \cite{berkes2003garch}, the quasi-MLE for $\btheta_0$ is defined as
\begin{align*}
    \hbtheta_n:= \arg \max_{\btheta \in \Theta} \sum_{t=1}^n  \widetilde{\ell}_t(\btheta), \quad 
    \widetilde{\ell}_t(\btheta) = -\log \widetilde{\sigma}_t(\btheta) - \frac{y_t^2}{2 \widetilde{\sigma}_t^2(\btheta)}.
\end{align*}
We further define $\ell_t(\btheta) = -\log \sigma_t(\btheta) - y_t^2/ (2 \sigma_t^2(\btheta))$.
To analyze the asymptotic properties of $\hbtheta_n$, consider the following assumptions:
\begin{assumption} \label{ass7}
    $\Theta$ is compact, and $\btheta_0$ lies in the interior of $\Theta$. 
\end{assumption}
\begin{assumption} \label{ass8}
    $\gamma(\bA_0) <0$, and $\forall \btheta \in \Theta$, $\sum_{j=1}^{q} \psi_j <1 $.
\end{assumption}
\begin{assumption} \label{ass9}
    $\{\eta_t^2\}$ is i.i.d. with $E \eta_t^2=1$, and $\kappa_\eta := E \eta_t^4 < \infty$.
\end{assumption}
\begin{assumption} \label{ass10}
    The GARCH$(p, q)$ representation is minimal, i.e., the polynomials $A(z) = \sum_{i=1}^p \phi_{0i} z^i$ and $B(z) = 1-  \sum_{j=1}^q \psi_{0j} z^j$ do not have common roots.
\end{assumption}

Next, we verify the hypotheses in Theorem \ref{thm1}.
For GARCH model, $E[y_t^2 | \mathbf{I}_{t-1}] = \sigma_t^2(\btheta_0)$, so model \eqref{garch} can be rewritten in the conditional-mean form:
\begin{align*}
    \bigg\{
	\begin{array}{l}
		y_t^2 = \sigma_t^2(\btheta_0) + \sigma_t^2(\btheta_0) (\eta_t^2 - 1), \\
		\sigma_{t}^{2}(\btheta_0)=\omega_0 + \sum_{i=1}^p \phi_{0i}  y_{t-i}^{2} + \sum_{j=1}^{q} \psi_{0j} \sigma_{t-j}^{2}(\btheta_0),
	\end{array} \bigg.
	~ t \in \mathbb{Z}.
\end{align*}
Denote $Y_t$ as the square of volatility, i.e., $Y_t = y_t^2$. Then the error term
\begin{align*}
    \varepsilon_t = Y_{t} - E[Y_{t} | \mathbf{I}_{t-1}] = \sigma_t^2(\btheta_0) (\eta_t^2 - 1),
\end{align*}
which is a martingale difference sequence with respect to $\mathcal{F}_{t-1}$, and
\begin{align*}
    f(\mathbf{I}_{t-1}, \btheta_0) = E[Y_{t} | \mathbf{I}_{t-1}] = \sigma_t^2(\btheta_0), \quad
    \mathbf{g}_{t}(\boldsymbol{\theta}_0)= \frac{ \partial \sigma_t^2(\btheta_0) }{\partial \boldsymbol{\theta} }.
\end{align*}

\begin{corollary} \label{corollary3}
  Under Assumptions \ref{ass7}--\ref{ass10}, the stationary GARCH($p, q$) process \eqref{garch} satisfies Assumptions \ref{ass1}--\ref{ass5} and condition \eqref{key condition} with quasi maximum likelihood estimation (MLE), which is
  \begin{align*}
	E \Big[ \varepsilon_t  w_{t-j}(\mathbf{x}) \mathbf{h}(Y_{t}, \mathbf{I}_{t-1}, \boldsymbol{\theta}_{0}) \Big]^{\prime} 
	=
	E\big[w_{t-j}(\mathbf{x}) \mathbf{g}_{t}(\boldsymbol{\theta}_0)\big]^{\prime} \mathbf{L}(\boldsymbol{\theta}_{0}), \quad \forall j \geq 1.
\end{align*}
Hence, if $\kappa_{ra} = 2 \kappa_{ov}$, by Theorems \ref{thm1}--\ref{thm-local H1}, the process $S_{n, w}$ satisfies
    \begin{equation*}
        S_{n, w} \Longrightarrow S_w^0, \quad  \text{in} ~ L_{2}(\Pi, v),
    \end{equation*}
    and the proposed test is consistent under the alternatives.
\end{corollary}

\begin{proof}

Here we verify the hypotheses for the sample-splitting-based generalized spectral tests for the stationary GARCH($p, q$) process \eqref{garch}.
Assumptions \ref{ass1}, \ref{ass4} are basic assumptions, and Assumptions \ref{ass2}--\ref{ass3} and \ref{ass5} are verified in the Section A.4 of Supplementary Material of \cite{davis2025sample}.
It remains to verify condition \eqref{key condition}.

Recall that $Y_t$ as the square of volatility, i.e., $Y_t = y_t^2$, and the error term
\begin{align*}
    \varepsilon_t = Y_{t} - E[Y_{t} | \mathbf{I}_{t-1}] = \sigma_t^2(\btheta_0) (\eta_t^2 - 1),
\end{align*}
is a martingale difference sequence with respect to $\mathcal{F}_{t-1}$, and
\begin{align*}
    f(\mathbf{I}_{t-1}, \btheta_0) = E[Y_{t} | \mathbf{I}_{t-1}] = \sigma_t^2(\btheta_0), \quad
    \mathbf{g}_{t}(\boldsymbol{\theta}_0)= \frac{ \partial \sigma_t^2(\btheta_0) }{\partial \boldsymbol{\theta} }.
\end{align*}
By Taylor expansion, Theorem 2.2 and (4.10)-(4.11) in \cite{francq2004maximum}, 
it can be shown that
\begin{align*}
    \sqrt{n} (\hbtheta_n - \btheta_0) 
    &= - \bigg( \frac{1}{n} \sum_{t=1}^n \frac{\partial^2}{\partial \btheta \partial \btheta' } \widetilde{\ell}_t(\btheta^*) \bigg)^{-1}
    \frac{1}{\sqrt{n}} \bigg( \sum_{t=1}^n \frac{\partial}{\partial \btheta} \widetilde{\ell}_t(\btheta_0) \bigg) \\
    &= - J^{-1} \frac{1}{\sqrt{n}} \sum_{t=1}^n \frac{\partial}{\partial \btheta} \ell_t(\btheta_0) + o_p(1)\\
    &\stackrel{d}{\longrightarrow} \mathcal{N}(0, (\kappa_\eta-1)J^{-1}),
\end{align*}
where $\btheta^*$ are between $\hbtheta_n$ and $\btheta_0$, and
\begin{align*}
    J := E_{\btheta_0}\bigg( \frac{\partial^2 \ell_t(\btheta_0)}{\partial \btheta \partial \btheta' } \bigg) = E_{\btheta_0}\bigg( \frac{1}{\sigma_t^4(\btheta_0)} \frac{\partial \sigma_t^2(\btheta_0)}{\partial \btheta } \frac{\partial \sigma_t^2(\btheta_0)}{\partial \btheta'} \bigg).
\end{align*}
It implies that
\begin{align*}
    &\mathbf{L}(\boldsymbol{\theta}_{0}) = (\kappa_\eta-1)J^{-1}, \\
    &\mathbf{h}(Y_{t}, \mathbf{I}_{t-1}, \boldsymbol{\theta}_{0}) = - J^{-1} \frac{\partial}{\partial \btheta} \ell_t(\btheta_0) 
    = - J^{-1} (1-\eta_t^2) \frac{1}{\sigma_t^2(\btheta_0)} \frac{\partial \sigma_t^2(\btheta_0)}{\partial \btheta}.
\end{align*}
Thus, it yields that
\begin{align*}
	E \Big[ \varepsilon_t  w_{t-j}(\mathbf{x}) \mathbf{h}(Y_{t}, \mathbf{I}_{t-1}, \boldsymbol{\theta}_{0}) \Big]^{\prime} 
	&= E \Big[ \sigma_t^2(\btheta_0) (\eta_t^2 - 1)  w_{t-j}(\mathbf{x}) J^{-1} (\eta_t^2-1) \frac{1}{\sigma_t^2(\btheta_0)} \frac{\partial \sigma_t^2(\btheta_0)}{\partial \btheta} \Big]^{\prime} \\
    &= E \Big[ (\eta_t^2 - 1)^2  w_{t-j}(\mathbf{x}) \frac{\partial \sigma_t^2(\btheta_0)}{\partial \btheta} \Big]^{\prime} J^{-1}  \\
	&= E\Big[ E \Big( (\eta_t^2 - 1)^2  w_{t-j}(\mathbf{x}) \frac{\partial \sigma_t^2(\btheta_0)}{\partial \btheta} \Big| \mathcal{F}_{t-1} \Big) \Big]^{\prime}  J^{-1} \\
	&= E\Big[ w_{t-j}(\mathbf{x}) \frac{\partial \sigma_t^2(\btheta_0)}{\partial \btheta} \Big]' E (\eta_t^2 - 1)^2 J^{-1} \\
	&=  E\Big[ w_{t-j}(\mathbf{x}) \frac{\partial \sigma_t^2(\btheta_0)}{\partial \btheta} \Big]' (\kappa_\eta-1)J^{-1}\\
    &= E\big[w_{t-j}(\mathbf{x}) \mathbf{g}_{t}(\boldsymbol{\theta}_0)\big]^{\prime} \mathbf{L}(\boldsymbol{\theta}_{0}) , \quad \forall j \geq 1,
\end{align*}
which verifies the condition \eqref{key condition}.
    
\end{proof}

\section{Additional Technical Details}

\subsection{Proof of Theorem \ref{thm1}}
Define $\tilde S_{n,w}(\boldsymbol{\eta},\hbtheta_{f_n})$ as the same process as $S_{n, w} (\boldsymbol{\eta}, \hbtheta_{f_n} )$ but with $\mathbf{I}_{t-1}$ replacing $\widehat{\mathbf{I}}_{t-1}$. Then, using the same argument as Lemma A.1 in \cite{escanciano2006goodness}, under Assumptions \ref{ass4}--\ref{ass5}, we   have that 
\begin{align*}
&E\bigl\| S_{n,w}(\boldsymbol{\eta},\hbtheta_{f_n}) - \tilde S_{n,w}(\boldsymbol{\eta},\hbtheta_{f_n}) \bigr\|^2
\\=&\sum_{j=1}^{l_n} \frac{1}{(j\pi)^2} n_j^{-1}
E\left(
\sum_{t=n-l_n+j}^n \bigl( f(\mathbf{I}_{t-1},\hbtheta_{f_n}) - f(\hat{\mathbf{I}}_{t-1},\hbtheta_{f_n}) \bigr)
w_{t-j}(\bf x)
\right)^2 \\
\leq& C \sum_{j=1}^{l_n} \frac{1}{(j\pi)^2} n_j^{-1}
\left(
\sum_{t=n-l_n+j}^n
\Bigl( E \sup_{\mathbf{\theta} \in \Theta}
\bigl( f(\mathbf{I}_{t-1},\mathbf{\theta}) - f(\hat{\mathbf{I}}_{t-1},\mathbf{\theta}) \bigr)^2 \Bigr)^{1/2}
\right)^2\to 0.
\end{align*}
This implies that 
\begin{align*}
    \big\|S_{n, w} (\boldsymbol{\eta}, \hbtheta_{f_n})-\widetilde{S}_{n, w} (\boldsymbol{\eta}, \hbtheta_{f_n})\big\|^{2} \xrightarrow{p} 0.
\end{align*}
Hence, with a bit abuse of notation,  we work with $S_{n, w} (\boldsymbol{\eta}, \hbtheta_{f_n} )$ in what follows.

By the Lagrange's mean value theorem, we have
\begin{equation} \label{eq-1}
S_{n, w}(\boldsymbol{\eta}, \hbtheta_{f_n})=S_{n, w}(\boldsymbol{\eta}, \boldsymbol{\theta}_{0})+\frac{\partial S_{n, w}(\boldsymbol{\eta}, \widetilde{\boldsymbol{\theta}}_{f_n})}{\partial \boldsymbol{\theta}^{\prime}}(\hbtheta_{f_n}-\boldsymbol{\theta}_{0}),
\end{equation}
where $\widetilde{\boldsymbol{\theta}}_{f_n}$ lies on the line segment between $\hbtheta_{f_n}$ and $\boldsymbol{\theta}_{0}$.

The following proof consists of four steps:
\begin{itemize}
    \item[(a)] We show that \begin{align} \label{eq-2}
    \bigg\|\frac{1}{\sqrt{l_n}} \frac{\partial S_{n, w}(\boldsymbol{\eta}, \widetilde{\boldsymbol{\theta}}_{f_n})}{\partial \boldsymbol{\theta}}+ \mathbf{G}_{w} (\boldsymbol{\eta}, \boldsymbol{\theta}_{0})  \bigg\| \xrightarrow{p} 0,
\end{align}
\item[(b)] Combined with (a),  Assumption \ref{ass3} and \eqref{eq-1}, \begin{equation}\label{eq-3}
    S_{n,w}(\boldsymbol{\eta},\hbtheta_{f_n})
=
S_{n, w}(\boldsymbol{\eta}, \boldsymbol{\theta}_{0}) - 
\mathbf{G}_{w}^{\prime} (\boldsymbol{\eta}, \boldsymbol{\theta}_{0})
\sqrt{\frac{l_n}{f_n}} \frac{1}{\sqrt{f_n}} \sum_{t=1}^{f_n} \mathbf{h}(Y_{t}, \mathbf{I}_{t-1}, \boldsymbol{\theta}_{0}) +
r_n(\boldsymbol{\eta}),
\end{equation} 
where $\|r_n\|=o_p(1)$.
\item[(c)] We show the weak convergence of \begin{align}   \label{eq-4}
    \Big( S_{n, w}(\boldsymbol{\eta}, \boldsymbol{\theta}_{0}), 
    \sqrt{\frac{l_n}{f_n}} \frac{1}{\sqrt{f_n}} \sum_{k=1}^{f_n} \mathbf{h}(Y_{k}, \mathbf{I}_{k-1}, \boldsymbol{\theta}_{0}) \Big)\Longrightarrow(S_w^0(\boldsymbol{\eta}), \sqrt{\kappa_{ra}}\boldsymbol{V}),
\end{align}
where  $\boldsymbol{V}$ is a Gaussian vector with mean $\mathbf{0}$ and variance-covariance matrix $\mathbf{L}(\boldsymbol{\theta}_{0})$, with 
$
\operatorname{cov}\bigl(S_w^0(\boldsymbol{\eta}), \mathbf{V}\bigr)
= \kappa_{ov}/\sqrt{\kappa_{ra}}\sum_{j=1}^{\infty}
E \bigl[
\varepsilon_t \, w_{t-j}(\bx)\, \mathbf{h}(Y_t, \mathbf{I}_{t-1}, \mathbf{\theta}_0)
\bigr]\Psi_j(\lambda).
$
\item[(d)] We show that if  condition \eqref{key condition} holds and $\kappa_{ra}=2\kappa_{ov}$, then 
$$
S_w^0(\boldsymbol{\eta}) - 
\mathbf{G}_{w}^{\prime} (\boldsymbol{\eta}, \boldsymbol{\theta}_{0})
\sqrt{\kappa_{ra}}\boldsymbol{V}\stackrel{d}{=}S_w^0(\boldsymbol{\eta}).
$$
\end{itemize}

\noindent\textit{Proof of (a)} Note that the process $S_{n, w}(\boldsymbol{\eta}, \hbtheta_{f_n})$ can be rewritten as
\begin{align*}
    S_{n, w}(\boldsymbol{\eta}, \hbtheta_{f_n}) 
    &=\sum_{j=1}^{l_n} n_{j}^{1 / 2} \bigg\{\frac{1}{n_{j}} \sum_{t=n-l_n+j}^{n} e_{t}(\hbtheta_{f_n}) w (\mathbf{Z}_{t-j}, \mathbf{x} ) \bigg\} \frac{\sqrt{2} \sin j \pi \lambda}{j \pi} \\
    & =\frac{1}{\sqrt{l_n}} \sum_{t=n-l_n+1}^{n} e_{t}(\hbtheta_{f_n}) \bigg\{ \sum_{j=1}^{t-(n-l_n)} \sqrt{\frac{l_n}{n_j}} w_{t-j}(\mathbf{x}) \Psi_j(\lambda) \bigg\}.
\end{align*}
It follows that
\begin{align*}
    \frac{1}{\sqrt{l_n}} \frac{\partial S_{n, w}(\boldsymbol{\eta}, \widetilde{\boldsymbol{\theta}}_{f_n})}{\partial \boldsymbol{\theta}} 
    & =\frac{1}{l_n} \sum_{t=n-l_n+1}^{n} \frac{\partial e_{t}(\widetilde{\boldsymbol{\theta}}_{f_n})}{\partial \boldsymbol{\theta}} \bigg\{ \sum_{j=1}^{t-(n-l_n)} \sqrt{\frac{l_n}{n_j}} w_{t-j}(\mathbf{x}) \Psi_j(\lambda) \bigg\} \\
    & = - \frac{1}{l_n} \sum_{t=n-l_n+1}^{n} \mathbf{g}_t(\widetilde{\boldsymbol{\theta}}_{f_n}) \bigg\{ \sum_{j=1}^{t-(n-l_n)} \sqrt{\frac{l_n}{n_j}} w_{t-j}(\mathbf{x}) \Psi_j(\lambda) \bigg\} \\
& =-\sum_{j=1}^{l_n} \frac{1}{l_n} \sum_{t=n-l_n+ j}^{n} l_n^{1 / 2} n_{j}^{-1 / 2} \mathbf{g}_{t}(\widetilde{\boldsymbol{\theta}}_{f_n}) w_{t-j}(\mathbf{x}) \Psi_{j}(\lambda) \\
& :=-\sum_{j=1}^{l_n} \mathbf{b}_{j, n}(\mathbf{x}, \widetilde{\boldsymbol{\theta}}_{f_n}) \Psi_{j}(\lambda),
\end{align*}
where $\mathbf{b}_{j, n}(\mathbf{x}, \widetilde{\boldsymbol{\theta}}_{f_n})=l_n^{-1} \sum_{t=n-l_n+j}^{n} l_n^{1 / 2} n_{j}^{-1 / 2} \mathbf{g}_{t}(\widetilde{\boldsymbol{\theta}}_{f_n}) w_{t-j}(\mathbf{x})$. 

Define $\mathbf{b}_{j}(\mathbf{x}, \boldsymbol{\theta}) = E[w_{t-j}(\mathbf{x}) \mathbf{g}_{t}(\btheta)]$, which is continuous in $\btheta$  under Assumption \ref{ass2}.   Furthermore,  define 
\begin{flalign*}
   R_n(\boldsymbol{\eta},\btheta)=&\frac{1}{\sqrt{l_n}} \frac{\partial S_{n, w}(\boldsymbol{\eta}, \btheta)}{\partial \boldsymbol{\theta}}+\sum_{j=1}^{l_n} \mathbf{b}_{j}(\mathbf{x}, \boldsymbol{\theta}_{0}) \Psi_{j}(\lambda) 
   \\=& -\sum_{j=1}^{l_n} [\mathbf{b}_{j, n}(\mathbf{x}, \btheta)-\mathbf{b}_{j}(\mathbf{x}, \boldsymbol{\theta})] \Psi_{j}(\lambda)-\sum_{j=1}^{l_n} [\mathbf{b}_{j}(\mathbf{x}, \boldsymbol{\theta})-\mathbf{b}_{j}(\mathbf{x}, \boldsymbol{\theta}_0)] \Psi_{j}(\lambda)
   \\=& R_{n,1}(\boldsymbol{\eta},\btheta)+R_{n,2}(\boldsymbol{\eta},\btheta).
\end{flalign*}
Therefore, to show \eqref{eq-2}, it suffices to show (1) $\|R_{n,1}(\boldsymbol{\eta},\widetilde{\btheta}_{f_n})\|=o_p(1)$,  (2) $\|R_{n,2}(\boldsymbol{\eta},\widetilde{\btheta}_{f_n})\|=o_p(1)$, and (3) $\|\sum_{j=l_n+1}^{\infty} \mathbf{b}_{j}(\mathbf{x}, \boldsymbol{\theta}_{0}) \Psi_{j}(\lambda) \|=o(1)$. Note that (3) is trivial by letting $n\to\infty$.

(1) For any fixed $K>0$, we have that,
\begin{flalign*}
   R_{n,1}(\boldsymbol{\eta},\btheta)
   &= -\sum_{j=1}^{K} [\mathbf{b}_{j, n}(\mathbf{x}, \btheta)-\mathbf{b}_{j}(\mathbf{x}, \boldsymbol{\theta})] \Psi_{j}(\lambda)-\sum_{j=K+1}^{l_n} [\mathbf{b}_{j, n}(\mathbf{x}, \btheta)-\mathbf{b}_{j}(\mathbf{x}, \boldsymbol{\theta})] \Psi_{j}(\lambda)\\
   &:= R_{n,11}(\boldsymbol{\eta},\btheta) +R_{n,12}(\boldsymbol{\eta},\btheta). 
\end{flalign*}
It  suffices to show that
$\sup_{\btheta\in\Theta}\|R_{n,11}(\boldsymbol{\eta},\btheta)\|=o_p(1)$ and 
$\sup_{\btheta\in\Theta}\|R_{n,12}(\boldsymbol{\eta},\btheta)\|=o_p(1)$.
Note that 
\begin{flalign*}
    \sup_{\btheta\in\Theta}\|R_{n,11}(\boldsymbol{\eta},\btheta)\|^2\leq &  \sum_{j=1}^K \frac{C}{j^2} \sup_{\btheta\in\Theta}\int_{\Upsilon}[\mathbf{b}_{j, n}(\mathbf{x}, \btheta)-\mathbf{b}_{j}(\mathbf{x}, \boldsymbol{\theta})]^2 W(d{\bf x})
    \\=& \sum_{j=1}^K \frac{C}{j^2} \sup_{\btheta\in\Theta}\int_{\Upsilon_c}[\mathbf{b}_{j, n}(\mathbf{x}, \btheta)-\mathbf{b}_{j}(\mathbf{x}, \boldsymbol{\theta})]^2 W(d{\bf x})\\&+\sum_{j=1}^K \frac{C}{j^2} \sup_{\btheta\in\Theta}\int_{\Upsilon\backslash\Upsilon_c}[\mathbf{b}_{j, n}(\mathbf{x}, \btheta)-\mathbf{b}_{j}(\mathbf{x}, \boldsymbol{\theta})]^2 W(d\bf{x})
    \\:=&\mathbf{A}_{1}+\mathbf{A}_2.
\end{flalign*}
Note by the boundedness of $w(\cdot)$, we have $|w_t({\bf x})|\leq C_w$ for some constant $C_w>0$. Furthermore, by Assumption \ref{ass2}, $|g_t(\btheta)|\leq \mathbf{M}_t$. Therefore, by Assumption \ref{ass4} and the uniform ergodic theorem (see e.g. Theorem 3.1 in \cite{ling2003asymptotic}), we have that 
\[
\sup_{\btheta\in\Theta}\sup_{\mathbf{x}\in\Upsilon_c}
|\mathbf{b}_{j,n}(\mathbf{x},\btheta)-\mathbf{b}_j(\mathbf{x},\btheta)|
\xrightarrow{p} 0.
\]
This implies that $\mathbf{A}_{1}=o_p(1)$. 

Moreover, it is easy to see that $
\sup_{\theta\in\Theta}\sup_{\bx\in\Upsilon}|\bb_{j,n}(\bx,\theta)|
\le
\frac{C_w}{n_j}\sum_{t=n-l_n+j}^n \mathbf{M}_t$.
Hence  under Assumption \ref{ass2}, we can show that  \begin{align*}
E\sup_{\theta\in\Theta}\int_{\Upsilon\setminus\Upsilon_c}
[\bb_{j,n}(\bx,\theta)-\bb_j(\bx,\theta)]^2 W(d\bx)
&\le
CW(\Upsilon\setminus\Upsilon_c)E\mathbf{M}_t^2.
\end{align*}
Since $W$ is a probability measure, for every $\varepsilon>0$ one can choose a
compact set $\Upsilon_c\subset\Upsilon$ such that $W(\Upsilon\setminus\Upsilon_c)<\varepsilon.$ Thus, by Chebyshev's inequality, we obtain that $\mathbf{A}_2=o_p(1)$.

Furthermore, using the orthogonality of $\Psi_j(\cdot)$ and $\Psi_{j'}(\cdot)$ for $j\neq j'$ and Cauchy-Schwarz inequality, we have
\begin{flalign*}
  \|R_{n,12}(\boldsymbol{\eta},\btheta)\|^2\leq&  2\sum_{j=K+1}^{l_n}\|\mathbf{b}_{j, n}( \mathbf{x}, \btheta)\Psi_{j}(\lambda)\|^2+2\sum_{j=K+1}^{l_n}\|\mathbf{b}_{j}( \mathbf{x}, \btheta)\Psi_{j}(\lambda)\|^2.
\end{flalign*}
Since $|g_t(\btheta)|\leq \mathbf{M}_t$ under Assumption \ref{ass2} and that $|w_{t-j}(\bx)|\leq C_w$, it is easy to see that 
$\|\mathbf{b}_{j, n}( \mathbf{x}, \btheta)\Psi_j(\lambda)\|\leq C_wj^{-1}l_n^{-1} \sum_{t=n-l_n+j}^n \mathbf{M}_t$, which does not depend on $\btheta$ and $\bx$. Hence, it follows
$$
E\sup_{\theta\in\Theta} \|R_{n,12}(\boldsymbol{\eta},\btheta)\|^2\leq E\mathbf{M}_t^2\sum_{j=K+1}^{l_n} j^{-2}.
$$ Using the Chebyshev's inequality, we have that $\sup_{\theta\in\Theta}\|R_{n,12}(\boldsymbol{\eta},\btheta)\| \stackrel{p}{\longrightarrow} 0$ by first letting $n\to\infty$ and then $K\to\infty$. The uniform argument then ensures that $\|R_{n,1}(\boldsymbol{\eta},\widetilde{\btheta}_{f_n})\|=o_p(1)$.

(2) It is easy to see that, as $n\to\infty$,
\begin{flalign*}
\|R_{n,2}(\boldsymbol{\eta},\widetilde{\btheta}_{f_n})\|^2\leq &2 \|\mathbf{G}_{w} (\boldsymbol{\eta}, \boldsymbol{\theta}_{0})-\mathbf{G}_{w} (\boldsymbol{\eta},\widetilde{\btheta}_{f_n})\|^2\\&+2 \sum_{j=l_n+1}^{\infty}\| [\mathbf{b}_{j}(\mathbf{x}, \widetilde{\btheta}_{f_n})-\mathbf{b}_{j}(\mathbf{x}, \boldsymbol{\theta}_0)] \Psi_{j}(\lambda)\|^2,
\end{flalign*}
where the second term on the RHS vanishes by letting $n\to\infty$ using  $|g_t(\btheta)|\leq \mathbf{M}_t$ under Assumption \ref{ass2} and that $|w_{t-j}(\bx)|\leq 1$.
Note by Assumption \ref{ass3}, $\widetilde{\btheta}_{f_n}$ is consistent, the continuous mapping theorem then implies that 
$\|R_{n,2}(\boldsymbol{\eta},\widetilde{\btheta}_{f_n})\|=o_p(1)$.

Therefore, summarizing (1) and (2) proves part (a).

\noindent\textit{Proof of (b)} is trivial by noting (a) and Assumption \ref{ass3}.

\noindent\textit{Proof of (c)} 
By Theorem 1 of \cite{escanciano2006generalized}, 
we have $$S_{n, w}(\cdot, \boldsymbol{\theta}_{0}) \Longrightarrow S_w^0(\cdot),$$
where $S_{w}^{0}(\cdot)$ is a Gaussian process in $L_{2}(\Pi, v)$ with mean $\mathbf{0}$ and covariance operator $C_{S_{w}^{0}}$ satisfying $\sigma_{h}^{2}=\langle C_{S_{w}^{0}}(h), h \rangle, \forall h \in L_{2}(\Pi, v)$, and $\sigma_{h}^{2}$ is as defined in \eqref{sigma_h}. This implies that $\{S_{n,w}(\cdot,\btheta_0)\}$ is tight in $L^2(\Pi,\nu)$. Furthermore, by Assumption \ref{ass3}, we have
$$(\sqrt{l_n}/f_n) \sum_{t=1}^{f_n} \mathbf{h}(Y_{t}, \mathbf{I}_{t-1}, \boldsymbol{\theta}_{0}) \xrightarrow{d} \sqrt{\kappa_{ra}} \boldsymbol{V}.$$ This implies that the second component is tight in $\mathbb{R}^p$.
Since $L^2(\Pi,\nu)$ is a separable Hilbert space and $\mathbb{R}^p$ is finite-dimensional, it follows that $L^2(\Pi,\nu)\times\mathbb{R}^p$ is a separable metric space. Therefore, the tightness of the two marginal sequences implies that the pair in \eqref{eq-4} is tight in $L^2(\Pi,\nu)\times\mathbb{R}^p$.

It remains to verify the convergence of finite-dimensional projections. Let
$g\in L_2(\Pi,\nu)$ and $\mathbf c\in\mathbb R^p$ be arbitrary. Consider
\begin{align} \label{thm1-proof-eq1}
\big\langle S_{n,w}(\cdot,\boldsymbol{\theta}_0),g\big\rangle
+
\mathbf c'
\sqrt{\frac{l_n}{f_n}}\frac{1}{\sqrt{f_n}}
\sum_{k=1}^{f_n}\mathbf h(Y_k,\mathbf I_{k-1},\boldsymbol{\theta}_0).
\end{align}
Using the representation
\begin{align*}
    \big\langle S_{n,w}(\cdot,\boldsymbol{\theta}_0),g\big\rangle
=
\frac{1}{\sqrt{l_n}}
\sum_{t=n-l_n+1}^{n}\varepsilon_t a_{n,t}(g),
\quad
a_{n,t}(g):=
\int_\Pi Q_{t,w}(\boldsymbol{\eta})g(\boldsymbol{\eta})\nu(d\boldsymbol{\eta}),
\end{align*}
where
\begin{align*}
Q_{t,w}(\boldsymbol{\eta})
:=
\sum_{j=1}^{t-(n-l_n)} \Big(\frac{l_n}{n_j}\Big)^{1/2}
w_{t-j}(\mathbf{x})\Psi_j(\lambda),
\quad
\boldsymbol{\eta}=(\lambda,\mathbf{x}')'\in \Pi,
\end{align*}
we can write the above quantity \eqref{thm1-proof-eq1} as $\sum_{t=1}^n \xi_{n,t}(g,\mathbf c)$, where
\begin{align*}
    \xi_{n,t}(g,\mathbf c)
=
\frac{\mathbbm 1(t\ge n-l_n+1)}{\sqrt{l_n}}\,
\varepsilon_t a_{n,t}(g)
+
\frac{\mathbbm 1(t\le f_n)}{\sqrt{f_n}}\sqrt{\frac{l_n}{f_n}}\,
\mathbf c'\mathbf h(Y_t,\mathbf I_{t-1},\boldsymbol{\theta}_0).
\end{align*}
Then $\{\xi_{n,t}(g,\mathbf c),\mathcal F_t\}_{t=1}^n$ is a martingale difference array,
as $\{\xi_{n,t}\}_{1 \leq t \leq n}$ is adapted to $\{\mathcal{F}_{t}\}_{1 \leq t \leq n}$, 
and $E[\varepsilon_t \{ \sum_{j=1}^{t-n+l_n} w_{t-j}(\mathbf{x}) \Psi_j(\lambda) \} | \mathcal{F}_{t-1}]=0$ 
and $E[\mathbf{h}(Y_{t}, \mathbf{I}_{t-1}, \btheta_{0}) | \mathcal{F}_{t-1}]=0$ by Assumption \ref{ass3}. 
Moreover, since $|w_{t-j}(\mathbf{x}) |\leq C_w$, we have $E[\varepsilon_t^2 \{ \sum_{j=1}^{t-n+l_n} w_{t-j}(\mathbf{x}) \Psi_j(\lambda) \}^2]< \infty$ and Assumption \ref{ass3}(b) ensures that  $E[ \mathbf{h}(Y_{t}, \mathbf{I}_{t-1}, \btheta_{0}) \mathbf{h}'(Y_{t}, \mathbf{I}_{t-1}, \btheta_{0})]$ exists and is positive definite. 
By Assumptions \ref{ass1}--\ref{ass4}, the same argument as in the proof of Theorem 1 of \cite{escanciano2006generalized} yields the conditional variance convergence and the Lindeberg condition. 
Hence, by the martingale central limit theorem,
\begin{align*}
    \sum_{t=1}^n \xi_{n,t}(g,\mathbf c)
\Longrightarrow
N \Big(0, \operatorname{Var} \big(\langle S_w^0,g\rangle+\mathbf c'\sqrt{\kappa_{ra}}\boldsymbol{V}\big)\Big).
\end{align*}
Note that
\begin{align*}
\operatorname{Var} \big(\langle S_w^0,g\rangle+\mathbf c'\sqrt{\kappa_{ra}}\boldsymbol{V}\big)
=
\operatorname{Var} (\langle S_w^0,g\rangle)
+\kappa_{ra}\mathbf c' \mathbf L(\btheta_0)\mathbf c
+2\mathbf c'\sqrt{\kappa_{ra}}\,
\operatorname{cov}(\langle S_w^0,g\rangle,\boldsymbol{V}).
\end{align*}

It remains to calculate the asymptotic covariance between the two components in \eqref{thm1-proof-eq1}.
For any fixed
$\boldsymbol{\eta}=(\lambda,\mathbf x')'\in\Pi$,
\begin{align}
    &\operatorname{cov}\Big(S_{n, w}(\boldsymbol{\eta}, \boldsymbol{\theta}_{0}),  
    \frac{\sqrt{l_n}}{f_n}  \sum_{k=1}^{f_n} \mathbf{h}(Y_{k}, \mathbf{I}_{k-1}, \boldsymbol{\theta}_{0}) \Big)  \nonumber \\
    =& \operatorname{cov} \bigg( \frac{1}{\sqrt{l_n}} \sum_{t=n-l_n+1}^{n} \varepsilon_t \bigg\{ \sum_{j=1}^{t-(n-l_n)} \sqrt{\frac{l_n}{n_j}} w_{t-j}(\mathbf{x}) \Psi_j(\lambda) \bigg\}, \frac{\sqrt{l_n}}{f_n}  \sum_{k=1}^{f_n} \mathbf{h}(Y_{k}, \mathbf{I}_{k-1}, \boldsymbol{\theta}_{0}) \bigg) \nonumber \\
    =& \frac{1}{f_n} \sum_{t=n-l_n+1}^{n}  \sum_{k=1}^{f_n} \operatorname{cov} \bigg( \varepsilon_t \bigg\{ \sum_{j=1}^{t-(n-l_n)} \sqrt{\frac{l_n}{n_j}} w_{t-j}(\mathbf{x}) \Psi_j(\lambda) \bigg\}, \mathbf{h}(Y_{k}, \mathbf{I}_{k-1}, \boldsymbol{\theta}_{0}) \bigg) \nonumber  \\
    =& \frac{1}{f_n} \sum_{t=n-l_n+1}^{n}  \sum_{k=1}^{f_n} E \bigg[ \varepsilon_t \bigg\{ \sum_{j=1}^{t-(n-l_n)} \sqrt{\frac{l_n}{n_j}} w_{t-j}(\mathbf{x}) \Psi_j(\lambda) \bigg\} \mathbf{h}(Y_{k}, \mathbf{I}_{k-1}, \boldsymbol{\theta}_{0}) \bigg], \label{eq-5}
\end{align}
where the last equation comes from the fact that $E[\mathbf{h}(Y_{t}, \mathbf{I}_{t-1}, \boldsymbol{\theta}_{0})]=\mathbf{0}$.
For $t>k$, the summand in \eqref{eq-5} equals $\mathbf{0}$ as $\varepsilon_t$ is a martingale difference sequence with respect to $\mathcal{F}_{t-1}$. 
For $t<k$,  under Assumption \ref{ass3} we also have
\begin{align*}
    &E \bigg[ \varepsilon_t \bigg\{ \sum_{j=1}^{t-(n-l_n)} \sqrt{\frac{l_n}{n_j}} w_{t-j}(\mathbf{x}) \Psi_j(\lambda) \bigg\} \mathbf{h}(Y_{k}, \mathbf{I}_{k-1}, \boldsymbol{\theta}_{0}) \bigg] \\
    =& E \bigg[  \varepsilon_t \bigg\{ \sum_{j=1}^{t-(n-l_n)} \sqrt{\frac{l_n}{n_j}} w_{t-j}(\mathbf{x}) \Psi_j(\lambda) \bigg\} E \bigg[ \mathbf{h}(Y_{k}, \mathbf{I}_{k-1}, \boldsymbol{\theta}_{0}) \bigg| \mathcal{F}_{k-1} \bigg] \bigg] =0.
\end{align*}
Thus, by the Stolz-Ces\`{a}ro theorem, it follows that
\begin{align*}
    &\operatorname{cov}\Big(S_{n, w}(\boldsymbol{\eta}, \boldsymbol{\theta}_{0}),  
    \frac{\sqrt{l_n}}{f_n}  \sum_{k=1}^{f_n} \mathbf{h}(Y_{k}, \mathbf{I}_{k-1}, \boldsymbol{\theta}_{0}) \Big)  \\
    =& \frac{1}{f_n} \sum_{t=n-l_n+1}^{n}  \sum_{k=1}^{f_n} \mathbbm{1}(t=k) E \bigg[ \varepsilon_t \bigg\{ \sum_{j=1}^{t-(n-l_n)} \sqrt{\frac{l_n}{n_j}} w_{t-j}(\mathbf{x}) \Psi_j(\lambda) \bigg\} \mathbf{h}(Y_{k}, \mathbf{I}_{k-1}, \boldsymbol{\theta}_{0}) \bigg] \\
    =& \frac{1}{f_n} \sum_{t=n-l_n+1}^{f_n}  E \bigg[ \varepsilon_t \bigg\{ \sum_{j=1}^{t-(n-l_n)} \sqrt{\frac{l_n}{n_j}} w_{t-j}(\mathbf{x}) \Psi_j(\lambda) \bigg\} \mathbf{h}(Y_{t}, \mathbf{I}_{t-1}, \boldsymbol{\theta}_{0}) \bigg] \\
    =& \frac{1}{f_n} \sum_{t=n-l_n+1}^{f_n} \sum_{j=1}^{t-(n-l_n)} \sqrt{(l_n/n_j)} E \Big[ \varepsilon_t  w_{t-j}(\mathbf{x}) \Psi_j(\lambda)  \mathbf{h}(Y_{t}, \mathbf{I}_{t-1}, \boldsymbol{\theta}_{0}) \Big]  \\
    \rightarrow& \kappa_{ov} \lim _{n \rightarrow \infty} \sum_{j=1}^{f_n-(n-l_n)} \sqrt{(l_n/n_j)} E \Big[ \varepsilon_t  w_{t-j}(\mathbf{x}) \mathbf{h}(Y_{t}, \mathbf{I}_{t-1}, \boldsymbol{\theta}_{0}) \Big] \Psi_j(\lambda)\\
    =&\kappa_{ov} \sum_{j=1}^{\infty} E \Big[ \varepsilon_t  w_{t-j}(\mathbf{x}) \mathbf{h}(Y_{t}, \mathbf{I}_{t-1}, \boldsymbol{\theta}_{0}) \Big] \Psi_j(\lambda).
\end{align*}
Hence,
$$
\operatorname{cov}\bigl(S_w^0(\boldsymbol{\eta}), \mathbf{V}\bigr)
= \frac{\kappa_{ov}}{\sqrt{\kappa_{ra}}}
\sum_{j=1}^{\infty}
E \bigl[
\varepsilon_t \, w_{t-j}(\bx)\, \mathbf{h}(Y_t, \mathbf{I}_{t-1}, \mathbf{\theta}_0)
\bigr]\Psi_j(\lambda).
$$
Specifically, under condition \eqref{key condition} that
$
E[ \varepsilon_t w_{t-j}(\mathbf x)\mathbf h(Y_t,\mathbf I_{t-1},\btheta_0)]
=
E [w_{t-j}(\mathbf x)\mathbf g_t(\btheta_0)]'\mathbf L(\btheta_0)$,
it follows that
\begin{align*}
\operatorname{cov} (S_w^0(\boldsymbol{\eta}),\boldsymbol{V})
=
\frac{\kappa_{ov}}{\sqrt{\kappa_{ra}}}
\sum_{j=1}^{\infty}
E\!\left[w_{t-j}(\mathbf x)\mathbf g_t(\btheta_0)\right]'
\mathbf L(\btheta_0)\Psi_j(\lambda)
=
\frac{\kappa_{ov}}{\sqrt{\kappa_{ra}}} \mathbf G_w'(\boldsymbol{\eta},\btheta_0)\mathbf L(\btheta_0).
\end{align*}

Therefore, every continuous linear functional of \eqref{eq-4} converges to the
corresponding linear functional of $(S_w^0(\cdot),\sqrt{\kappa_{ra}}\boldsymbol{V})$.
Together with joint tightness, we conclude (c).

\noindent\textit{Proof of (d)} Combining (a)-(c) and Slutsky's theorem, we obtain
\begin{align*}
S_{n,w}(\cdot,\hbtheta_{f_n})
\Longrightarrow
S_w^0(\cdot)-\sqrt{\kappa_{ra}} \mathbf{G}_{w}^{\prime} (\cdot, \boldsymbol{\theta}_{0}) \boldsymbol{V},
\qquad \text{in }L_2(\Pi,\nu),
\end{align*}
which is a Gaussian process with zero mean. 
Since $S_w^0(\cdot)$ is also a Gaussian process with zero mean,
it suffices to  compare the covariance bilinear forms induced by all continuous
linear functionals. 
 Let $g_1,g_2\in L_2(\Pi,\nu)$ and define
\[
\mathbf a_i
=\int_\Pi
\mathbf G_w(\boldsymbol{\eta},\boldsymbol{\theta}_0)
g_i^c(\boldsymbol{\eta})\,d\nu(\boldsymbol{\eta}),
\qquad i=1,2 .
\]
Then
$
\big\langle \mathbf G_w'(\cdot,\boldsymbol{\theta}_0)\boldsymbol{V},g_i\big\rangle
=\mathbf a_i'\boldsymbol{V} .
$
By (c), we have that \[
\operatorname{Cov}\{\langle S_w^0,g_i\rangle,\boldsymbol{V}\}
=\frac{\kappa_{ov}} {\sqrt{\kappa_{ra}}}\,
\mathbf a_i'\mathbf L(\boldsymbol{\theta}_0),
\qquad i=1,2.
\]
Hence, using $\operatorname{Var}(\boldsymbol{V})=\mathbf L(\boldsymbol{\theta}_0)$, it can be shown that 
\begin{align*}
&\operatorname{Cov}\!\left(
\big\langle S_w^0-\sqrt{\kappa_{ra}}\mathbf G_w'\boldsymbol{V},g_1\big\rangle,
\big\langle S_w^0-\sqrt{\kappa_{ra}}\mathbf G_w'\boldsymbol{V},g_2\big\rangle
\right)\\
&=\operatorname{Cov}\{\langle S_w^0,g_1\rangle,\langle S_w^0,g_2\rangle\}
+(\kappa_{ra}-2\kappa_{ov})
\mathbf a_1'\mathbf L(\boldsymbol{\theta}_0)\mathbf a_2 .
\end{align*}
When $\kappa_{ra}=2\kappa_{ov}$, the second term vanishes for every
$g_1,g_2\in L_2(\Pi,\nu)$. Thus the two centered Gaussian random elements
have the same covariance operator, and therefore
\[
S_w^0(\cdot)
-\sqrt{\kappa_{ra}}\mathbf G_w'(\cdot,\boldsymbol{\theta}_0)\boldsymbol{V}
\stackrel{d}{=}S_w^0(\cdot),
\]
and hence
\[
S_{n,w}(\cdot,\hbtheta_{f_n})
\Longrightarrow S_w^0(\cdot)
\quad\text{in }L_2(\Pi,\nu).
\]
The convergence of $D_{n,w}^2(\hbtheta_{f_n})$ follows from the continuous
mapping theorem. This completes the proof.

\hfill $\square$

\subsection{Proof of Theorem \ref{thm-global H1}}

Similar to the proof of Theorem \ref{thm1}, by Lemma A.1 in \cite{escanciano2006goodness}, we replace the unobserved information set $I_{t-1}$ by the observed one $\widehat I_{t-1}$, and work with the whole conditioning set $I_{t-1}$ throughout the proof for simplicity.

By Lemma 1 of \cite{escanciano2006generalized}, and note that
\begin{align*}
    S_{n, w} (\boldsymbol{\eta}, \hbtheta_{f_n} ) 
    =\sum_{j=1}^{l_n} n_{j}^{1 / 2} \widehat{\gamma}_{j, w} (\mathbf{x}, \hbtheta_{f_n} ) \frac{\sqrt{2} \sin j \pi \lambda}{j \pi},
\end{align*}
it is sufficient to prove that $\widehat{\gamma}_{j, w} (\mathbf{x}, \hbtheta_{f_n})$ satisfies (i) and (ii) of Lemma 1.
Recall that under $H_{a}$, $Y_{t}=f (\mathbf{I}_{t-1}, \boldsymbol{\theta}_{*})+a_{t}+\varepsilon_{t}$,
where $\{a_{t}\}$ is strictly stationary and ergodic, with $Ea_{1}^2<\infty$, and for each $t \in \mathbb{Z}$, $a_{t}$ is $\mathcal{F}_{t-1}$-measurable. 
It follows that
\begin{align} \label{thmH1-proof-eq1}
  \begin{split}
    \widehat{\gamma}_{j, w} (\mathbf{x}, \hbtheta_{f_n} )
    &=\frac{1}{n_{j}} \sum_{t=n-l_n+j}^{n} e_{t}(\hbtheta_{f_n}) w_{t-j} (\mathbf{x} ) \\
    &=\frac{1}{n_{j}} \sum_{t=n-l_n+j}^{n} \Big\{ a_{t}+\varepsilon_{t} + f (\mathbf{I}_{t-1}, \boldsymbol{\theta}_{*}) - f (\mathbf{I}_{t-1},\hbtheta_{f_n}) \Big\} w_{t-j} (\mathbf{x} )\\
    &=\frac{1}{n_j}\sum_{t=n-l_n+j}^{n} \big\{a_t + \varepsilon_t \big\} w_{t-j}(\mathbf{x}) \\
    & \quad + \frac{1}{n_j}\sum_{t=n-l_n+j}^{n}
    \Big\{ f(\mathbf I_{t-1},\boldsymbol{\theta}_*)
    -f(\mathbf I_{t-1},\hbtheta_{f_n})
    \Big\}w_{t-j}(\mathbf{x}) \\
    &:= A_{1n,j}(\mathbf{x}) + A_{2n,j}(\mathbf{x}),  \quad j \geq 1.
  \end{split}
\end{align}

For the first term, since $\{a_t\}$ is strictly stationary and ergodic with
$E|a_1|^2<\infty$, Assumption \ref{ass4} yields that
\begin{align*}
\sup_{\mathbf{x}\in\Upsilon_c}
\big|
A_{1n,j}(\mathbf{x})-\varsigma_j(\mathbf{x})
\big|
&=
\sup_{\mathbf{x}\in\Upsilon_c}
\Bigg|
\frac{1}{n_j}\sum_{t=n-l_n+j}^{n} (a_t + \varepsilon_t) w_{t-j}(\mathbf{x})
-
E[a_t w_{t-j}(\mathbf{x})]
\Bigg|
=o_p(1),
\end{align*}
where
\[
\varsigma_j(\mathbf{x}):=E[a_t w_{t-j}(\mathbf{x})].
\]

Next, we deal with $A_{2n,j}(\mathbf{x})$. By the mean value theorem, for each
$t=n-l_n+j,\ldots,n$, there exists
$\widetilde{\boldsymbol{\theta}}_{n,t}$ such that
\begin{align*}
f(\mathbf I_{t-1},\boldsymbol{\theta}_*)
-f(\mathbf I_{t-1},\hbtheta_{f_n})
=
-\mathbf g_t(\widetilde{\boldsymbol{\theta}}_{f_n,t})'
(\hbtheta_{f_n}-\boldsymbol{\theta}_*),
\end{align*}
where
$|\widetilde{\boldsymbol{\theta}}_{f_n,t}-\boldsymbol{\theta}_*|
\le
|\hbtheta_{f_n}-\boldsymbol{\theta}_*|
=o_p(1)$.
Hence,
\begin{align} \label{thmH1-proof-eq2}
\begin{split}
    A_{2n,j}(\mathbf{x})
&=
-(\hbtheta_{f_n}-\boldsymbol{\theta}_*)'
\frac{1}{n_j}\sum_{t=n-l_n+j}^{n}
\mathbf g_t(\widetilde{\boldsymbol{\theta}}_{f_n,t})w_{t-j}(\mathbf{x}) \\
&=
-(\hbtheta_{f_n}-\boldsymbol{\theta}_*)'
\frac{1}{n_j}\sum_{t=n-l_n+j}^{n}
\mathbf g_t(\boldsymbol{\theta}_*)w_{t-j}(\mathbf{x}) \\
&\quad
-(\hbtheta_{f_n}-\boldsymbol{\theta}_*)'
\frac{1}{n_j}\sum_{t=n-l_n+j}^{n}
\Big\{
\mathbf g_t(\widetilde{\boldsymbol{\theta}}_{f_n,t})
-
\mathbf g_t(\boldsymbol{\theta}_*)
\Big\}
w_{t-j}(\mathbf{x}) \\
&:= A_{21n,j}(\mathbf{x})+A_{22n,j}(\mathbf{x}).
\end{split}
\end{align}

For $A_{21n,j}(\mathbf{x})$, we have
\begin{align*}
\sup_{\mathbf{x}\in\Upsilon_c} \big|A_{21n,j}(\mathbf{x})\big|
&\le
|\hbtheta_{f_n}-\boldsymbol{\theta}_*|
\sup_{\mathbf{x}\in\Upsilon_c}
\Bigg|
\frac{1}{n_j}\sum_{t=n-l_n+j}^{n}
\mathbf g_t(\boldsymbol{\theta}_*)w_{t-j}(\mathbf{x})
\Bigg| \\
&\le
|\hbtheta_{f_n}-\boldsymbol{\theta}_*|
\Bigg[
\sup_{\mathbf{x}\in\Upsilon_c}
\Bigg|
\frac{1}{n_j}\sum_{t=n-l_n+j}^{n}
\mathbf g_t(\boldsymbol{\theta}_*)w_{t-j}(\mathbf{x})
-
E\big\{\mathbf g_t(\boldsymbol{\theta}_*)w_{t-j}(\mathbf{x})\big\}
\Bigg| \\
&\quad
+
\sup_{\mathbf{x}\in\Upsilon_c}
\left|
E\big\{\mathbf g_t(\boldsymbol{\theta}_*)w_{t-j}(\mathbf{x})\big\}
\right|
\Bigg].
\end{align*}
By Assumptions \ref{ass3}--\ref{ass4}, the fact that $\{\mathbf g_t(\boldsymbol{\theta}_*)\}$ is strictly stationary and ergodic, the first term inside the brackets is
$o_p(1)$, while the second term is finite since $|w(\cdot)| \leq 1$ is uniformly bounded and $E\|\mathbf g_t(\boldsymbol{\theta}_*)\|<\infty$. Therefore,
\begin{align}
    \label{thmH1-proof-eq3}
\sup_{\mathbf{x}\in\Upsilon_c} \big|A_{21n,j}(\mathbf{x})\big|
=o_p(1).
\end{align}

Next consider $A_{22n,j}(\mathbf{x})$. Since
$|\widetilde{\boldsymbol{\theta}}_{f_n,t}-\boldsymbol{\theta}_*|
\le |\hbtheta_{f_n}-\boldsymbol{\theta}_*|=o_p(1)$ and
$\mathbf g_t(\boldsymbol{\theta})$ is continuous in $\boldsymbol{\theta}$, we have
\begin{align*}
    \mathbf g_t(\widetilde{\boldsymbol{\theta}}_{f_n,t})
-
\mathbf g_t(\boldsymbol{\theta}_*)
\xrightarrow{p} 0,
\qquad \text{for each } t.
\end{align*}
By Assumptions \ref{ass2}--\ref{ass4} and the boundedness of
$w(\cdot)$, the dominated convergence theorem yields that
\begin{align*}
\sup_{\mathbf{x}\in\Upsilon_c}
\bigg|
\frac{1}{n_j}\sum_{t=n-l_n+j}^{n}
\Big\{
\mathbf g_t(\widetilde{\boldsymbol{\theta}}_{f_n,t})
-
\mathbf g_t(\boldsymbol{\theta}_*)
\Big\}
w_{t-j}(\mathbf{x})
\bigg|
=o_p(1).
\end{align*}
It follows that
\begin{align}
     \label{thmH1-proof-eq4}
    \sup_{\mathbf{x}\in\Upsilon_c}|A_{22n,j}(\mathbf{x})|
=o_p(1).
\end{align}
Combining \eqref{thmH1-proof-eq2}--\eqref{thmH1-proof-eq4} yields
\[
\sup_{\mathbf{x}\in\Upsilon_c}|A_{2n,j}(\mathbf{x})|
=o_p(1).
\]

Therefore, back in \eqref{thmH1-proof-eq1}, for each fixed $j\ge 1$,
\begin{align*}
\sup_{\mathbf{x}\in\Upsilon_c}
\left|
\widehat{\gamma}_{j,w}(\mathbf{x},\hbtheta_{f_n})
-\varsigma_j(\mathbf{x})
\right|
=o_p(1).
\end{align*}
This proves condition (ii) of Lemma 1 in \cite{escanciano2006generalized}.
Condition (i), which controls the contribution of large lags in
$L_2(\Pi,\nu)$, follows from the finite second moment conditions in
Assumptions \ref{ass1}--\ref{ass2}, $E|a_1|^2<\infty$, the boundedness of
$w(\cdot)$, and the square summability of $\|\Psi_j\|_{L_2[0,1]}^2
=O(j^{-2})$.  
Thus, Lemma 1 yields that
\begin{align*}
    n_j^{-1 / 2} S_{n, w}(\cdot, \hbtheta_{f_n}) 
    \stackrel{p}{\longrightarrow} 
    L_{w}(\cdot),
    \qquad \text{in } L_2(\Pi,\nu),
\end{align*}
where $L_{w}(\boldsymbol{\eta})=\sum_{j=1}^{\infty} \varsigma_{j}(\mathbf{x}) \Psi_{j}(\lambda)$, $\varsigma_{j}(\mathbf{x})=E[a_{t} w_{t-j}(\mathbf{x})]$, which completes the proof.
\hfill $\square$

\subsection{Proof of Corollary \ref{corollary1}}

For $\{a_t\}\in\Xi$, there exists some $j\ge 1$ such that
$\varsigma_j(\mathbf{x})\neq 0$ on a subset of $\Upsilon$ with positive W-measure.
By Lemma 1 of \cite{escanciano2006goodness}, it is equivalent to the fact that under $H_a$, there exists some $j\ge 1$ such that
$E(a_t\mid \mathbf Z_{t-j})\neq 0$ on a set with positive probability.

Since $W(\cdot)$ is absolutely continuous with respect to the Lebesgue measure, we have
\[
\int_{\Upsilon_c}\varsigma_j^2(\mathbf{x})\,W(d\mathbf{x})>0
\]
for some $j\ge 1$. 
Consequently,
\begin{align*}
l_n^{-1}D_{n,w}^2(\hbtheta_{f_n})
\stackrel{p}{\longrightarrow}
\sum_{j=1}^{\infty} \frac{1}{(j \pi)^2}  \int_{\Upsilon_c}\varsigma_j^2(\mathbf{x})\,W(d\mathbf{x}),
\end{align*}
where the limit is strictly positive under $H_a$. Therefore,
\begin{align*}
    D_{n,w}^2(\hbtheta_{f_n}) \stackrel{p}{\longrightarrow}\infty,
\end{align*}
which means the test based on $D_{n,w}^2$ is consistent against any fixed alternative in
$\Xi$. This completes the proof.
\hfill $\square$

\subsection{Proof of Theorem \ref{thm-local H1}}

Similar to the proof of Theorem \ref{thm1}, we replace the unobserved information set $I_{t-1}$ by the observed one $\widehat I_{t-1}$.
Recall that for the local alternative sequence $H_{a,n}$, $Y_{t, n}=f(\mathbf{I}_{t-1}, \boldsymbol{\theta}_{0})+ a_{t}/\sqrt{l_n} +\varepsilon_{t}$,
where $\{a_{t}\}$ is strictly stationary and ergodic, with $Ea_{1}^2<\infty$, and for each $t \in \mathbb{Z}$, $a_{t}$ is $\mathcal{F}_{t-1}$-measurable. 
Then 
$$e_t(\btheta_0)=Y_{t,n}-f(\mathbf{I}_{t-1}, \boldsymbol{\theta}_{0}) =\varepsilon_t+\frac{a_t}{\sqrt{l_n}}.$$
Let
\begin{align*}
Q_{t,w}(\boldsymbol{\eta})
:=
\sum_{j=1}^{t-(n-l_n)} \Big(\frac{l_n}{n_j}\Big)^{1/2}
w_{t-j}(\mathbf{x})\Psi_j(\lambda),
\quad
\boldsymbol{\eta}=(\lambda,\mathbf{x}')'\in \Pi.
\end{align*}
Then
\begin{align*}
S_{n,w}(\boldsymbol{\eta},\hbtheta_{f_n})
=
\frac{1}{\sqrt{l_n}}
\sum_{t=n-l_n+1}^n e_t(\hbtheta_{f_n}) Q_{t,w}(\boldsymbol{\eta}),
\end{align*}

As in the proof of Theorem \ref{thm1}, by the Lagrange's mean value theorem, we have
\begin{equation} \label{thm3:eq-1}
S_{n, w}(\boldsymbol{\eta}, \hbtheta_{f_n}) = S_{n, w}(\boldsymbol{\eta}, \boldsymbol{\theta}_{0})+ \frac{\partial S_{n, w}(\boldsymbol{\eta}, \widetilde{\boldsymbol{\theta}}_{f_n})}{\partial \boldsymbol{\theta}^{\prime}}(\hbtheta_{f_n}-\boldsymbol{\theta}_{0})
\end{equation}
where $\widetilde{\boldsymbol{\theta}}_{f_n}$ is a mean value satisfying $|\widetilde{\boldsymbol{\theta}}_{f_n}-\boldsymbol{\theta}_{0}| \leq|\hbtheta_{f_n}-\boldsymbol{\theta}_{0}|=o_p(1)$ a.s. 
Now define, exactly as in the proof of Theorem \ref{thm1} that
\begin{align*}
    R_n(\boldsymbol{\theta})
:=
\frac{1}{\sqrt{l_n}} \frac{\partial S_{n, w}(\boldsymbol{\eta}, \widetilde{\boldsymbol{\theta}}_{f_n})}{\partial \boldsymbol{\theta}}
+
\sum_{j=1}^{l_n} \mathbf{b}_{j}(\mathbf{x}, \boldsymbol{\theta}_{0}) \Psi_{j}(\lambda) ,
\quad
\mathbf{b}_{j}(\mathbf{x}, \boldsymbol{\theta}) = E[w_{t-j}(\mathbf{x}) \mathbf{g}_{t}(\btheta)].
\end{align*}
By the same argument used in proving \eqref{eq-2} in Theorem \ref{thm1}, it can be shown that
\begin{align*}
    \|R_n(\widetilde{\boldsymbol{\theta}}_{f_n})\|=o_p(1).
\end{align*}
Therefore,
\begin{align}  \label{thm3:eq-2}
  \begin{split}
      S_{n, w}(\boldsymbol{\eta}, \hbtheta_{f_n})
    &=S_{n, w}(\boldsymbol{\eta}, \boldsymbol{\theta}_{0}) + 
    \Big\{ -\sum_{j=1}^{l_n} \mathbf{b}_{j}(\mathbf{x}, \boldsymbol{\theta}_{0}) \Psi_{j}(\lambda) \Big\}^{\prime}
    \sqrt{l_n} (\hbtheta_{f_n}-\boldsymbol{\theta}_{0}) + o_p(1) \\
    &=S_{n, w}(\boldsymbol{\eta}, \boldsymbol{\theta}_{0}) - 
    \mathbf{G}_{w}^{\prime} (\boldsymbol{\eta}, \boldsymbol{\theta}_{0})
    \sqrt{l_n} (\hbtheta_{f_n}-\boldsymbol{\theta}_{0}) + o_p(1).
  \end{split}
\end{align}

Next, we decompose $S_{n, w}(\boldsymbol{\eta}, \boldsymbol{\theta}_{0})$ under $H_{a,n}$ \eqref{local H1} as
\begin{align}
S_{n, w}(\boldsymbol{\eta}, \boldsymbol{\theta}_{0})
&=
\frac{1}{\sqrt{l_n}}
\sum_{t=n-l_n+1}^n
\Big(\varepsilon_t+\frac{a_t}{\sqrt{l_n}}\Big)Q_{t,w}(\boldsymbol{\eta})
\notag\\
&=
\frac{1}{\sqrt{l_n}}\sum_{t=n-l_n+1}^n \varepsilon_t Q_{t,w}(\boldsymbol{\eta})
+
\frac{1}{l_n}\sum_{t=n-l_n+1}^n a_t Q_{t,w}(\boldsymbol{\eta})
\notag\\
&:= S_{n,w}^{(0)}(\boldsymbol{\eta})+L_{n,w}(\boldsymbol{\eta}).
\label{thm3:eq-3}
\end{align}

For $S_{n,w}^{(0)}(\boldsymbol{\eta})$,
by Theorem 1 of \cite{escanciano2006generalized}, 
we have $S_{n,w}^{(0)}(\cdot) \Rightarrow S_w^0(\cdot)$,
where $S_{w}^{0}(\cdot)$ is defined in Theorem \ref{thm1}.
For $L_{n,w}(\boldsymbol{\eta})$, note that
\begin{align*}
    L_{n,w}(\boldsymbol{\eta}) =
\sum_{j=1}^{l_n}
\Big(
\frac{1}{\sqrt{n_j l_n}}
\sum_{t=n-l_n+j}^n a_t w_{t-j}(\mathbf{x})
\Big)\Psi_j(\lambda).
\end{align*}
Since $\{a_t\}$ is strictly stationary and ergodic with
$E|a_1|^2<\infty$, Assumption \ref{ass4} yields that for $j \geq 1$,
\begin{align*}
\sup_{\mathbf{x}\in\Upsilon_c}
\bigg|
\frac{1}{\sqrt{n_j l_n}}\sum_{t=n-l_n+j}^{n} a_t w_{t-j}(\mathbf{x})
-
E[a_t w_{t-j}(\mathbf{x})]
\bigg|
=o_p(1).
\end{align*}
Indeed, the left-hand term is
\[
\sqrt{\frac{n_j}{l_n}}
\left\{ \frac{1}{n_j}\sum_{t=n-l_n+j}^{n}a_t w_{t-j}(\mathbf{x}) \right\},
\]
and for each fixed $j$, $n_j/l_n\to 1$.
Hence, by similar arguments as in the proof of Theorem \ref{thm-global H1} and Lemma 1 in \cite{escanciano2006generalized}, we obtain
\begin{align}
L_{n,w}(\cdot)\stackrel{p}{\longrightarrow}L_w(\cdot)
\qquad\text{in }L_2(\Pi,\nu),
\label{thm3:eq-4}
\end{align}
where $L_w(\boldsymbol{\eta})=\sum_{j=1}^\infty \varsigma_j(\mathbf{x})\Psi_j(\lambda)$,
$\varsigma_j(\mathbf{x})=E[a_t w_{t-j}(\mathbf{x})]$.

It remains to handle $S_{n,w}^{(0)}(\cdot)$ in \eqref{thm3:eq-3} jointly with the estimation effect in \eqref{thm3:eq-2}. 
By Assumption \ref{H1-ass1}, we have
\begin{align}
\sqrt{l_n}(\hbtheta_{f_n}-\btheta_0)
=
\sqrt{\frac{l_n}{f_n}} \boldsymbol{\xi}_a
+
\sqrt{\frac{l_n}{f_n}}
\frac{1}{\sqrt{f_n}}
\sum_{t=1}^{f_n} \mathbf{h} (Y_{t}, \mathbf{I}_{t-1}, \boldsymbol{\theta}_{0})
+
o_p(1).
\label{thm3:eq-5}
\end{align}
Combining \eqref{thm3:eq-2} -- \eqref{thm3:eq-5} and Theorem \ref{thm1}, we get
\begin{align}
\begin{split}
    S_{n,w}(\boldsymbol{\eta},\hbtheta_{f_n})
=&
\bigg\{ S_{n,w}^{(0)}(\boldsymbol{\eta}) -
\mathbf{G}_{w}^{\prime} (\boldsymbol{\eta}, \boldsymbol{\theta}_{0})
\sqrt{\frac{l_n}{f_n}} \frac{1}{\sqrt{f_n}}
\sum_{t=1}^{f_n} \mathbf{h} (Y_{t}, \mathbf{I}_{t-1}, \boldsymbol{\theta}_{0})
\bigg\}\\
&
+ L_{n,w}(\boldsymbol{\eta}) -
\mathbf{G}_{w}^{\prime} (\boldsymbol{\eta}, \boldsymbol{\theta}_{0}) \sqrt{\frac{l_n}{f_n}} \xi_a
+ r_n(\boldsymbol{\eta}),
\end{split}
\label{thm3:eq-6}
\end{align}
where $\|r_n\|=o_p(1)$, in ${L_2(\Pi,\nu)}$.

Now, the bracketed term in \eqref{thm3:eq-6} is exactly the same centered Gaussian part as in the proof of Theorem \ref{thm1}. 
Thus, using the same joint convergence argument as in
Theorem \ref{thm1}, together with condition \eqref{key condition} and $\kappa_{ra}=2\kappa_{ov}$,
we obtain
\begin{align}
S_{n,w}^{(0)}(\boldsymbol{\eta}) -
\mathbf{G}_{w}^{\prime} (\boldsymbol{\eta}, \boldsymbol{\theta}_{0})
\sqrt{\frac{l_n}{f_n}} \frac{1}{\sqrt{f_n}}
\sum_{t=1}^{f_n} \mathbf{h} (Y_{t}, \mathbf{I}_{t-1}, \boldsymbol{\theta}_{0})
\Longrightarrow
S_w^0(\boldsymbol{\eta}),
\quad\text{in} ~L_2(\Pi,\nu).
\label{thm3:eq-7}
\end{align}
Combining \eqref{thm3:eq-4}, \eqref{thm3:eq-6}, and \eqref{thm3:eq-7}, together with similar arguments of weak convergence in Theorem \ref{thm1} and Slutsky's theorem,
we can show that
\begin{align*}
    S_{n,w}(\cdot,\hbtheta_{f_n})
\Longrightarrow
S_w^0(\cdot)+L_w(\cdot)- \sqrt{\kappa_{ra}} \mathbf{G}_w'(\cdot,\btheta_0) \boldsymbol{\xi}_a,
\quad \text{in }L_2(\Pi,\nu),
\end{align*}
where $\kappa_{ra} \equiv \lim_{n \rightarrow \infty} l_n/f_n$.
This proves the first assertion.
The asymptotic distribution of $D_{n, w}^{2}(\hbtheta_{f_n})$ follows directly using the continuous mapping theorem and Assumption \ref{ass5}.
This completes the proof.
\hfill $\square$

\subsection{Proof of Theorem \ref{thm2}}

As in the proof of Theorem \ref{thm1}, we replace the unobserved information set $I_{t-1}$ by the observed one $\widehat I_{t-1}$, and work with the whole conditioning set $I_{t-1}$ throughout the proof.

Let
$Q_{t,w}(\boldsymbol{\eta})
:=
\sum_{j=1}^{t-(n-l_n)} \big(\frac{l_n}{n_j}\big)^{1/2}
w_{t-j}(\mathbf{x})\Psi_j(\lambda)$,
$\boldsymbol{\eta}=(\lambda,\mathbf{x}')'\in \Pi$.
Then
\begin{align*}
S_{n,w}^*(\boldsymbol{\eta},\hbtheta_{f_n})
=
\frac{1}{\sqrt{l_n}}
\sum_{t=n-l_n+1}^n e_t(\hbtheta_{f_n})V_tQ_{t,w}(\boldsymbol{\eta}),
\end{align*}
where $\{V_t\}$ is a sequence of i.i.d. random variables with zero mean, unit variance,
bounded support, and independent of the sample. Also define
\begin{align*}
S_{n,w,L}^*(\boldsymbol{\eta})
:=
\frac{1}{\sqrt{l_n}}
\sum_{t=n-l_n+1}^n e_t(\boldsymbol{\theta}_*)V_tQ_{t,w}(\boldsymbol{\eta}).
\end{align*}

By the Lagrange mean value theorem,
\begin{align*}
S_{n,w}^*(\boldsymbol{\eta},\hbtheta_{f_n})
=
S_{n,w,L}^*(\boldsymbol{\eta})
+
\frac{\partial S_{n,w}^*(\boldsymbol{\eta},\widetilde{\boldsymbol{\theta}}_{f_n})}
{\partial \boldsymbol{\theta}'}
(\hbtheta_{f_n}-\boldsymbol{\theta}_*),
\end{align*}
where $\widetilde{\boldsymbol{\theta}}_{f_n}$ is a mean value satisfying
$|\widetilde{\boldsymbol{\theta}}_{f_n}-\boldsymbol{\theta}^*|
\le
|\hbtheta_{f_n}-\boldsymbol{\theta}^*|
=o_p(1)$.
Note that
\begin{align*}
\frac{1}{\sqrt{l_n}}
\frac{\partial S_{n,w}^*(\boldsymbol{\eta},\widetilde{\boldsymbol{\theta}}_{f_n})}
{\partial \boldsymbol{\theta}}
&=
\frac{1}{l_n}
\sum_{t=n-l_n+1}^n
\frac{\partial e_t(\widetilde{\boldsymbol{\theta}}_{f_n})}
{\partial \boldsymbol{\theta}}
V_tQ_{t,w}(\boldsymbol{\eta}) \\
&=
-
\frac{1}{l_n}
\sum_{t=n-l_n+1}^n
\mathbf g_t(\widetilde{\boldsymbol{\theta}}_{f_n})
V_tQ_{t,w}(\boldsymbol{\eta}) \\
&=
-
\sum_{j=1}^{l_n}
\mathbf b_{j,n}^*(\mathbf x,\widetilde{\boldsymbol{\theta}}_{f_n})\Psi_j(\lambda),
\end{align*}
where $\mathbf b_{j,n}^*(\mathbf x,\boldsymbol{\theta})
:=
l_n^{-1}
\sum_{t=n-l_n+j}^n
l_n^{1/2}n_j^{-1/2}\mathbf g_t(\boldsymbol{\theta})V_tw_{t-j}(\mathbf x)$.
It follows that
\begin{align*}
S_{n,w}^*(\boldsymbol{\eta},\hbtheta_{f_n})
&=
S_{n,w,L}^*(\boldsymbol{\eta})
+
\bigg\{ -\sum_{j=1}^{l_n}
\mathbf b_{j,n}^*(\mathbf x,\widetilde{\boldsymbol{\theta}}_{f_n})\Psi_j(\lambda) \bigg\}
(\hbtheta_{f_n}-\boldsymbol{\theta}_*) \\
&:= S_{n,w,L}^*(\boldsymbol{\eta})
+ \sqrt{l_n} R_n^*(\widetilde{\boldsymbol{\theta}}_{f_n})' 
(\hbtheta_{f_n}-\boldsymbol{\theta}_*).
\end{align*}

We aim to show that
\begin{align}
\bigl\|R_n^*(\widetilde{\boldsymbol{\theta}}_{f_n})\bigr\|
=o_{p}(1). \label{eq:boot-derivative-small}
\end{align}
To this end, write
\begin{align*}
R_n^*(\boldsymbol{\theta})
&=
-\sum_{j=1}^{l_n}
\Bigl[
\mathbf b_{j,n}^*(\mathbf x,\boldsymbol{\theta})
-
E\{\mathbf b_{j,n}^*(\mathbf x,\boldsymbol{\theta})\}
\Bigr]\Psi_j(\lambda)
-
\sum_{j=1}^{l_n}
E\{\mathbf b_{j,n}^*(\mathbf x,\boldsymbol{\theta})\}\Psi_j(\lambda) \\
&:= R_{n,1}^*(\boldsymbol{\theta})+R_{n,2}^*(\boldsymbol{\theta}).
\end{align*}
Since $E(V_t)=0$, we have $E\{\mathbf b_{j,n}^*(\mathbf x,\boldsymbol{\theta})\}=0$,  $\forall j$,
and thus $R_{n,2}^*(\boldsymbol{\theta})\equiv 0$.
It remains to study $R_{n,1}^*(\boldsymbol{\theta})$. For any fixed $K>0$, write
\begin{align*}
R_{n,1}^*(\boldsymbol{\theta})
&=
-\sum_{j=1}^{K}
\mathbf b_{j,n}^*(\mathbf x,\boldsymbol{\theta})\Psi_j(\lambda)
-\sum_{j=K+1}^{l_n}
\mathbf b_{j,n}^*(\mathbf x,\boldsymbol{\theta})\Psi_j(\lambda) \\
&:= R_{n,11}^*(\boldsymbol{\theta})+R_{n,12}^*(\boldsymbol{\theta}).
\end{align*}
By Assumption \ref{ass2} and \ref{ass4}, the uniform ergodic theorem, the bounded support of $V_t$, and the same truncation argument
as in the proof of Theorem \ref{thm1}, we have, for each fixed $K$,
\begin{align*}
\sup_{\theta\in\Theta} \bigl\|R_{n,11}^*(\boldsymbol{\theta})\bigr\|
=o_{p}(1),
\end{align*}
while the tail term satisfies
\begin{align*}
E \sup_{\theta\in\Theta}  \bigl\|R_{n,12}^*(\boldsymbol{\theta})\bigr\|^2
\le
C\sum_{j=K+1}^{\infty}j^{-2}.
\end{align*}
Using the Chebyshev's inequality, we have that $\sup_{\theta\in\Theta}\|R_{n,12}^*(\btheta)\| \xrightarrow{p} 0$ by first letting $n\to\infty$ and then $K\to\infty$. The uniform argument then ensures that $\|R_{n,1}(\widetilde{\btheta}_{f_n})\|=o_p(1)$.
Hence \eqref{eq:boot-derivative-small} follows, and
consequently,
\begin{align*}
S_{n,w}^*(\boldsymbol{\eta},\hbtheta_{f_n})
=
S_{n,w,L}^*(\boldsymbol{\eta})
+
\sqrt{l_n}\,R_n^*(\widetilde{\boldsymbol{\theta}}_{f_n})'
(\hbtheta_{f_n}-\boldsymbol{\theta}^*)
=
S_{n,w,L}^*(\boldsymbol{\eta})+o_{p}(1)
\end{align*}
in $L_2(\Pi,\nu)$, combining \eqref{eq:boot-derivative-small} and Assumption \ref{ass3}.

It remains to study $S_{n,w,L}^*$. For any $h\in L_2(\Pi,\nu)$,
\begin{align*}
\langle S_{n,w,L}^*,h\rangle
=
\frac{1}{\sqrt{l_n}}
\sum_{t=n-l_n+1}^n a_{n,t}(h)V_t,
\end{align*}
where
\begin{align*}
a_{n,t}(h)
:=
e_t(\boldsymbol{\theta}_*)
\int_{\Pi}Q_{t,w}(\boldsymbol{\eta})h(\boldsymbol{\eta})\,\nu(d\boldsymbol{\eta}).
\end{align*}
Conditionally on the sample, $\{V_t\}$ are i.i.d. centered random variables with unit
variance and bounded support. Hence, by the conditional Lindeberg central limit theorem,
together with the same covariance calculation as in Theorem \ref{thm1}, we obtain
\begin{align*}
\langle S_{n,w,L}^*,h\rangle
\stackrel{d}{\longrightarrow}
N(0,\sigma_h^2(\boldsymbol{\theta}_*)),
\quad \text{a.s.}
\end{align*}
The convergence of the finite-dimensional distributions follows from the Cramér--Wold
device. Tightness in $L_2(\Pi,\nu)$ follows exactly as in the proof of Theorem 1 in \cite{escanciano2006generalized},
using the orthogonality of $\{\Psi_j\}$ and Theorem 2.5.2 of \cite{van1996weak}. Therefore,
\begin{align*}
S_{n,w,L} \underset{*}{\Longrightarrow} \widetilde{S}_{w}
\quad \text{in }L_2(\Pi,\nu) \quad \text{a.s.},
\end{align*}
where $\widetilde{S}_{w}$ is the same Gaussian process as in Theorem \ref{thm1}, with
$\boldsymbol{\theta}_*$ replacing $\boldsymbol{\theta}_0$.

Combining the above, we conclude that
\begin{align*}
S_{n,w}^*(\cdot,\hbtheta_{f_n})
\underset{*}{\Longrightarrow} \widetilde{S}_{w}
\quad \text{in }L_2(\Pi,\nu) \quad \text{a.s.}
\end{align*}
This completes the proof.
\hfill $\square$

\bibliographystyle{apalike}
\bibliography{bibliography.bib}

\end{document}